%% file: bounded.tex
\pdfoutput=1

\newif\ifdraft\draftfalse
\newif\iffull\fulltrue

\newcommand{\LONGVERSION}{\iffull{}appendix\else{}extended version\fi}

\documentclass[acmsmall,capitalize,authorversion,screen,amsthm]{acmart}

\usepackage[utf8]{inputenc}
\usepackage[T1]{fontenc}
\usepackage{xspace}
\usepackage{mathtools,stmaryrd}
\usepackage[bb=boondox]{mathalfa}
\usepackage{graphicx,xspace,xcolor}
\usepackage{paralist,enumitem}
\usepackage{nicefrac}
\usepackage{todonotes}
\usepackage{amsthm}

\input{prelude}
\newcommand{\glauber}{\mathsf{glauber}}
\newcommand{\sgm}{\mathsf{sgm}}
\newcommand{\PopDyn}{\mathsf{popdyn}}

\DeclareMathOperator{\NEIGHBORS}{\mathcal{N}}
\newcommand{\neighbors}[1]{\NEIGHBORS_{#1}}

\DeclareMathOperator{\VALID}{\mathcal{V}}
\newcommand{\valid}[1]{\VALID_{#1}}

\title{Proving Expected Sensitivity of Probabilistic Programs}
\iffull\else
\titlenote{The extended version is available at \url{https://arxiv.org/abs/1708.02537}.}
\fi

\author{Gilles Barthe}
\affiliation{%
  \institution{Imdea Software Institute}
  \city{Madrid}
\country{Spain}}
\author{Thomas Espitau}
\affiliation{%
  \institution{Sorbonne Universités, UPMC}
  \city{Paris}
\country{France}}
\author{Benjamin Gr\'egoire}
\affiliation{%
  \institution{Inria Sophia Antipolis--Méditerranée}
  \city{Nice}
\country{France}}
\author{Justin Hsu}
\orcid{0000-0002-8953-7060}
\affiliation{%
  \institution{University College London}
  \city{London}
\country{UK}}
\author{Pierre-Yves Strub}
\affiliation{%
  \institution{École Polytechnique}
  \city{Paris}
\country{France}}



\makeatletter\if@ACM@journal\makeatother
\acmJournal{PACMPL}
\acmVolume{2}
\acmNumber{POPL}
\acmArticle{57}
\acmYear{2018}
\acmMonth{1}
\acmDOI{10.1145/3158145}
\startPage{1}
\else\makeatother
\acmConference[PL'17]{ACM SIGPLAN Conference on Programming Languages}{January 01--03, 2017}{New York, NY, USA}
\acmYear{2017}
\acmISBN{978-x-xxxx-xxxx-x/YY/MM}
\acmDOI{10.1145/nnnnnnn.nnnnnnn}
\startPage{1}
\fi

\setcopyright{rightsretained}
\begin{document}

\begin{abstract}
  \emph{Program sensitivity}, also known as \emph{Lipschitz continuity},
  describes how small changes in a program's input lead to bounded changes in
  the output. We propose an average notion of program sensitivity for
  probabilistic programs---\emph{expected sensitivity}---that averages a
  distance function over a probabilistic coupling of two output
  distributions from two similar inputs. By varying the distance,
  expected sensitivity recovers useful notions of probabilistic function
  sensitivity, including stability of machine learning algorithms
  and convergence of Markov chains.

  Furthermore, expected sensitivity satisfies clean compositional
  properties and is amenable to formal verification. We develop a
  relational program logic called \SYSTEM for proving expected
  sensitivity properties. Our logic features two key ideas. First,
  relational pre-conditions and post-conditions are expressed using
  \emph{distances}, a real-valued generalization
  of typical boolean-valued (relational) assertions. Second, judgments
  are interpreted in terms of \emph{expectation coupling}, a novel,
  quantitative generalization of probabilistic couplings which
  supports compositional reasoning.

  We demonstrate our logic on examples beyond the reach of prior relational
  logics. Our main example formalizes uniform stability of the stochastic
  gradient method.
  Furthermore, we prove rapid mixing for a probabilistic model of population
  dynamics. We also extend our logic with a transitivity principle for
  expectation couplings to capture the \emph{path coupling} proof
  technique by \citet{bubley1997path}, and formalize
  rapid mixing of the Glauber dynamics from
  statistical physics.
\end{abstract}

\maketitle

\section{Introduction}
\emph{Sensitivity} is a
fundamental property in mathematics and computer science, describing
how small changes in inputs can affect outputs. Formally,
the sensitivity of a function $g:A \to B$ is defined relative to two
metrics $\preexp_{A}$ and $\preexp_{B}$ on $A$ and $B$
respectively. We say that $f$ is $\alpha$-sensitive if for every two inputs $x_1$ and
$x_2$, the outputs are a bounded distance apart: $\preexp_{B} (g(x_1),g(x_2)) \leq \alpha \cdot
\preexp_{A}(x_1,x_2)$. Bounded sensitivity plays a central
role in many other fields, motivating 
broad range verification methods for bounding program sensitivity.

We consider \emph{expected} (or average) sensitivity, a natural
generalization of sensitivity for the probabilistic setting, and
develop a program logic for proving expected sensitivity of
probabilistic programs. We work with a mild generalization of
sensitivity, called $f$-\emph{sensitivity}.
Formally, let $\preexp: A \times A\to \RR^+$
and $\postexp:B \times B\to \RR^+$ be two distances, and let $f$ be a non-negative
affine function of the form $z \mapsto \alpha \cdot z + \beta$ with $\alpha,
\beta$ non-negative. We say that $g:A \to B$ is $f$-sensitive iff for for
every two inputs $x_1$ and $x_2$, $\preexp_{B} (g(x_1),g(x_2)) \leq
f(\preexp_{A}(x_1,x_2))$. Taking $f$ to be affine will allow $f$-sensitivity to
compose cleanly in the probabilistic case, while still being expressive enough
to model multiplicative and additive bounds on the output distance in terms of
the input distance.

\subsection{Expected Sensitivity}
Let us now consider the case where $g$ is probabilistic,
i.e., $g:A \to \distr(B)$. Since $g$ produces distributions over $B$
rather than elements of $B$, we have a choice of what output distance to take.
One possibility is to allow arbitrary distances between distributions; however,
such distances can be complex and difficult
to reason about. We consider an alternative approach: lifting a distance $\preexp_B$ on
elements to a distance on distributions by averaging $\preexp_B$ over some
distribution $\mu$ on pairs $B \times B$. For any two nearby inputs of $g$ leading to
output distributions $\mu_1$ and $\mu_2$, we require this distribution $\mu$ to
model $\mu_1$ and $\mu_2$ in a probabilistic sense; namely,
its first and second marginals must be equal to
$\mu_1$ and $\mu_2$. Such a distribution $\mu$ is known
as a \emph{probabilistic coupling} of $\mu_1$ and $\mu_2$ (we refer the reader
to \citet{Lindvall02} and \citet{Thorisson00} for overviews of the rich theory
of probabilistic couplings).

Formally, a probabilistic function $g$ is \emph{expected} $f$-\emph{sensitive}
if for every two inputs $x_1$ and $x_2$, there exists a coupling $\mu$ of
$g(x_1)$ and $g(x_2)$, such that
\begin{equation} \label{eq:intro:def}
  \E {(y_1, y_2) \sim \mu} { \preexp_{B}(y_1, y_2) }
  \leq f(\preexp_{A}(x_1, x_2)) .
\end{equation}
The left-hand side is the \emph{expected value} of the function $\preexp_B$ over
$\mu$ (the average distance between pairs drawn from $\mu$), inspired by the
\emph{Wasserstein} metric, a well-studied distance on
distributions in the theory of optimal transport~\citep{Villani08}. Our notion
of expected sensitivity has several appealing features. First, it is quite
general---we can capture many probabilistic notions of sensitivity by varying
the distance.
% For two basic examples:
%
\begin{example}[Average sensitivity]
When the outputs $(y_1, y_2)$ are numbers, a natural notion of sensitivity
bounds the difference between \emph{average} outputs in terms of the distance between inputs
$(x_1, x_2)$. Taking the distance $\preexp_{B}(y_1, y_2) \triangleq |y_1 -
y_2|$, expected $f$-sensitivity implies
\[
  \left|\E {y_1 \sim \mu_1} {y_1} - \E {y_2 \sim \mu_2} {y_2}\right|
  \leq f(\preexp_{A}(x_1, x_2)).
\]
In other words, the two output distributions $\mu_1$ and $\mu_2$ have similar
averages when the inputs $(x_1, x_2)$ are close. This type of bound can
imply \emph{algorithmic stability}, a useful property for machine learning
algorithms~\citep{BousquetE02}.
\end{example}

\begin{example}[Probabilistic sensitivity]
Suppose that the output distance $\preexp_{B}$ is bounded away from
zero: $\preexp_{B}(y_1, y_2) < 1$ iff $y_1 = y_2$; for instance,
$\preexp_{B}$ could be an integer-valued metric. Then,
expected $f$-sensitivity implies
\[
  \left| \Pr_{y_1 \sim \mu_1} [y_1 \in E] - \Pr_{y_2 \sim \mu_2} [y_2 \in E] \right|
  \leq f(\preexp_{A}(x_1, x_2))
\]
for every subset of outputs $E$. This inequality shows that the distributions
$\mu_1$ and $\mu_2$ are close in a pointwise sense, and can imply that two
sequences of distributions converge to one another.
\end{example}

Another appealing feature of expected sensitivity is closure under composition:
the sequential (Kleisli) composition of an $f$-sensitive function with an
$f'$-sensitive function yields an $f' \circ f$-sensitive function. As we will
see, this property makes expected sensitivity a good target for formal
verification.

\subsection{Expected Sensitivity from Expectation Couplings}
To bound expected distance, it suffices to find a coupling of the output
distributions and show that the expected distance is sufficiently small. In
general, there are multiple probabilistic couplings between any two
distributions, leading to different expected distances.

To better reason about couplings and their expected distances, we develop a
quantitative generalization of probabilistic coupling that captures
\cref{eq:intro:def}; namely, if a
distribution $\mu$ on pairs satisfies the bound for expected
sensitivity, we call $\mu$ an \emph{expectation coupling} of $\mu_1$
and $\mu_2$ with respect to $\preexp_{B}$ and $\delta$, where $\delta=
f(\preexp_{A}(x_1, x_2))$. We show that expectation couplings satisfy several
natural properties, including closure under sequential composition and
transitivity.

\subsection{\SYSTEM: A Program Logic for Expected Sensitivity Bounds}

By leveraging these principles, we can bound expected sensitivity by
compositionally building an expectation coupling between output
distributions from pairs of nearby inputs. Concretely, we develop a relational
program logic \SYSTEM to construct expectation couplings between pairs of
programs. \SYSTEM judgments have the form
\[
  \eprhl{s_1}{s_2}{\pre; \preexp}{\post; \postexp}{f} ,
\]
where $s_1$ and $s_2$ are probabilistic imperative programs---often,
the same program---the pre- and post-conditions $\pre,\post: \mem
\times \mem \to \BB$ are relational assertions over pairs of memories,
$\preexp, \postexp : \mem \times \mem \to \RR^+$ are non-negative
distances on memories, and $f(z) = \alpha \cdot z + \beta$ is a
non-negative affine function with $\alpha,\beta\geq 0$. \SYSTEM
judgments state that for any pair of related input memories $(m_1,
m_2)$ satisfying the pre-condition $\pre$, there exists an expectation
coupling $\mu$ of the output distributions such that
all pairs of output memories $(m_1', m_2')$ in the support of $\mu$
(i.e., with positive probability) satisfy the
post-condition $\post$, and the expected distance is bounded:
\[
  \E {(m_1', m_2') \sim \mu} {\postexp(m_1', m_2')}
  \leq f (\preexp(m_1, m_2)) = \alpha \cdot \preexp(m_1, m_2) + \beta .
\]
We call $f$ a \emph{distance transformer}, as it bounds
the (average) \emph{post-distance} $\postexp$ in terms of the
\emph{pre-distance} $\preexp$. When $s_1$ and $s_2$ are the same program $s$,
for instance, a \SYSTEM judgment establishes
an expected sensitivity bound for $s$.

We give a rich Hoare-style proof system for \SYSTEM, internalizing various
composition properties of expectation couplings. For instance, given two
judgments
\[
  \eprhl{s_1}{s_2}{\pre; \preexp}{\Xi; \postexp}{f}
  \quad\text{and}\quad
  \eprhl{s_1'}{s_2'}{\Xi; \postexp}{\post; \postexpz}{f'} ,
\]
the sequential composition rule in \SYSTEM yields
\[
  \eprhl{s_1;s_1'}{s_2;s_2'}{\pre; \preexp}{\post; \postexpz}{f' \circ f} .
\]
Note that the pre- and post-conditions and the distances compose naturally,
while the distance transformers combine smoothly by function composition. As a
result, we can reason about the sequential composition of two programs by
building an expectation coupling for each.

\subsection{Applications}
We illustrate the expressiveness of our proof system on several novel examples.

\paragraph*{Stability of Stochastic Gradient Method.}
In machine learning, \emph{stability}~\citep{BousquetE02,ElisseeffEP05}
measures how changes in the training set influence the quality of an
algorithm's prediction. A stable algorithm does not
depend too much on the particular training set, so that its performance on the
training set generalizes to new, unseen examples; in other words,
it does not overfit. Recently, \citet{HardtRS16} show a quantitative stability
bound for the Stochastic
Gradient Method (SGM), a widely used optimization algorithm for
training in machine learning.  We verify their result
for several variants of SGM within our logic, contributing to
the expanding field of formal verification for machine
learning algorithms~\citep{KatzBDJK17,HuangKWW17,SelsamLD17}.

\paragraph*{Rapid Mixing for Population Dynamics.}
Randomized algorithms are a useful tool for modeling biological and social
phenomena (see, e.g., \citet{jansen2013analyzing}). They can be used to
analyze population dynamics, both in the infinite population setting where
evolution is \emph{deterministic}, and in the finite population setting where
evolution can be \emph{stochastic}. We formally analyze a variant of the
\emph{RSM} (Replication-Selection-Mutate) model, which captures the
evolution of an unstructured, asexual haploid population (see, e.g.,
\citet{HartlC}). Recently, a series of papers prove rapid mixing of the RSM
model under mild conditions \citep{dixit2012finite,Vishnoi15,PanageasSV16}.
We formally verify rapid mixing in a simplified setting, where the evolution
function is strictly contractive. This example relies on an advanced proof rule
internalizing the \emph{maximal} coupling of two multinomial distributions; in
some sense, the coupling that minimizes the expected distance between samples.

\subsection{Extension: Path Coupling}

Once we have set the core logic, we extend the rules to capture more
advanced reasoning about expectation couplings. We consider the
\emph{path coupling} method due to \citet{bubley1997path}, a
theoretical tool for building couplings on Markov chains. Let $\Phi$ be a binary
relation and
suppose that the state space of the Markov chain is equipped with a
\emph{path metric} $\preexp_\Phi$, i.e., the distance between
two elements is the length of the shortest $\Phi$-path between them.
We say that two states are \emph{adjacent} if their distance is
$1$. The main idea of path coupling is that if we can give a coupling
for the distributions starting from
neighboring states, then we can combine these pieces to give a
coupling for the distributions started from \emph{any} two
states. More concretely, if for every two initial states at distance $1$
under $\preexp_\Phi$ there is an expectation coupling of the output
distributions with expected distance at most $\gamma$, then for every
two initial states at distance $k$ under $\preexp_\Phi$, path
couplings gives an expectation coupling with expected distance at most
$k \cdot \gamma$.

From a logical point of view, path coupling is a transitivity principle for
expectation couplings: given a coupling for inputs related by $\Phi$, we get
a coupling for inputs related by $\Phi^k$. In \SYSTEM, we internalize
this principle by a structural transitivity rule, allowing a family
of relational judgments to be chained together. We formally prove rapid mixing for
a classical example called the Glauber dynamics, a Markov chain for drawing
approximately uniform samplings from the proper colorings of a
graph~\citep{bubley1997path}.

\subsection{Outline and Core Contributions}
After illustrating our approach on a simple example (\cref{sec:motiv}) and
reviewing mathematical preliminaries, we present the following contributions.
\begin{itemize}
  \item A novel abstraction called expectation couplings for reasoning about
    probabilistic sensitivity, supporting natural composition properties
    (\cref{sec:prelim}).
  \item A probabilistic relational program logic \SYSTEM for constructing
    expectation couplings, along with a proof of soundness (\cref{sec:logic}).
  \item A formal proof of uniform stability for two versions of the Stochastic
    Gradient Method, relying on proof rules to perform probabilistic case
    analysis (\cref{sec:sgm}).
  \item A formal proof of rapid mixing for a Markov chain simulating population
    evolution, relying on a proof rule internalizing the \emph{maximal} coupling
    of two multinomial distributions (\cref{sec:population}).
  \item A formal proof of rapid mixing for the Glauber dynamics from statistical
    physics, relying on an advanced proof rule internalizing the path coupling
    principle (\cref{sec:pathcoupling}).  
\end{itemize}
We have implemented our logic in the \textsf{EasyCrypt}~\citep{BartheDGKSS13}, a
general-purpose proof assistant for reasoning about probabilistic programs, and
machine-checked our main examples (\cref{sec:implem}). We conclude by surveying
related work (\cref{sec:rw}) and presenting promising future directions
(\cref{sec:conclusion}).

\section{Stability of Stochastic Gradient Method} \label{sec:motiv}

To give a taste of our approach, let's consider a property from
machine learning. In a typical learning setting, we have a space of possible
\emph{examples} $Z$, a \emph{parameter space} $\RR^d$, and a \emph{loss
  function} $\ell : Z \to \RR^d \to [0, 1]$. An algorithm $A$ takes a finite set
$S \in Z^n$ of \emph{training examples}---assumed to be drawn independently from
some unknown distribution $\mathcal{D}$---and produces parameters $w \in \RR^d$
aiming to minimize
the expected loss of $\ell(-, w)$ on a fresh sample from $\mathcal{D}$.
When the algorithm is randomized, we think of
$A : Z^n \to \distr(\RR^d)$ as mapping the training examples to a distribution
over parameters.

In order to minimize the loss on the true distribution $\mathcal{D}$,
a natural idea is to use parameters that minimize the average error on the
available training set. When the loss function $\ell$
is well-behaved this optimization problem, known as \emph{empirical risk
  minimization}, can be solved efficiently. However there is
no guarantee that these parameters \emph{generalize} to the true
distribution---even if they have low loss on the training set, they may induce
high loss on fresh samples from the true distribution. Roughly speaking, the
algorithm may select parameters that are too specific to the inputs,
\emph{overfitting} to the training set.

To combat overfitting, \citet{BousquetE02} considered a technical property of
the learning algorithm: the algorithm should produce similar outputs when
executed on any two training sets that differ in a single example, so that the
output does not depend too much on any single training example.

\begin{definition*}[\citet{BousquetE02}]
  Let $A : Z^n \to \distr(\RR^d)$ be an algorithm for some loss function
  $\ell : Z \to \RR^d \to [0, 1]$. $A$ is said to be $\epsilon$-\emph{uniformly
  stable} if for all input sets $S, S' \in
  Z^n$ that differ in a single element,\footnote{%
    In other words, $S$ and $S'$ have the same cardinality and their symmetric
    difference contains exactly two elements.}
  we have
  \[
    \E{w \sim A(S)}{\ell(z,w)} - \E{w \sim A(S')}{\ell(z, w)} \leq \epsilon
  \]
  for all $z \in Z$, where $\E{x \sim \mu}{f(x)}$ denotes the expected value of
  $f(x)$ when $x$ is sampled from $\mu$.
\end{definition*}

By the following observation, $\epsilon$-uniform stability follows from
an expected sensitivity condition, taking the distance on the
input space $Z^n$ to be the number of differing elements in $(S, S')$, and the
distance on output parameters to be the difference between losses $\ell$ on any
single example.

\begin{fact*}
  For every pair of training sets $S, S' \in Z^n$ that differ in a single
  element, suppose there exists a joint distribution $\mu(S, S') \in
  \distr(\RR^d \times \RR^d)$ such that $\pi_1(\mu) = A(S)$ and $\pi_2(\mu) =
  A(S')$, where $\pi_1, \pi_2 : \distr(\RR^d \times \RR^d) \to \distr(\RR^d)$
  give the first and second marginals. If
  \[
    \E{(w, w') \sim \mu(S, S')}{ |\ell(z, w) - \ell(z, w')| } \leq \epsilon
  \]
  for every $z \in Z$, then $A$ is $\epsilon$-uniformly stable.
\end{fact*}

The joint distribution $\mu(S, S')$ is an example of a \emph{expectation
coupling} of $A(S)$ and $A(S')$, where $|\ell(z, w) - \ell(z, w')|$ is viewed as
a \emph{distance} on pairs of parameters $(w, w')$. If we take the distance on
training sets $Z^n$ to be the symmetric distance (the number of differing
elements between training sets), $\epsilon$-uniform stability follows from
expected sensitivity of the function $A$. To prove stability, then, we will
establish expected sensitivity by (i) finding an expectation coupling and (ii)
reasoning about the expected value of the distance function. Our logic \SYSTEM
is designed to handle both tasks.

To demonstrate, we will show $\epsilon$-stability
of the \emph{Stochastic Gradient Method} (SGM), following recent work by
\citet{HardtRS16}. SGM is a simple and classical optimization algorithm commonly
used in machine learning. Typically, the parameter space is $\RR^d$ (i.e.,
the algorithm learns $d$ real parameters). SGM maintains parameters
$w$ and iteratively updates $w$ to reduce the loss. Each iteration,
SGM selects a uniformly random example $z$ from the input training set $S$ and
computes the \emph{gradient} vector
% \footnote{%
%   Roughly speaking, the gradient acts a multidimensional analogue of the usual
%   derivative.}
$g$ of the function $\ell(z, -) : \RR^d \to [0, 1]$ evaluated at
$w$---this vector indicates the direction to move $w$ to decrease the
loss. Then, SGM updates $w$ to step along $g$. After running $T$ iterations, the
algorithm returns final parameters. We can implement SGM in an
imperative language as follows.
\newcommand{\codesgm}{
\[
  \begin{array}{l}
    w \iass w_0; \\
    t \iass 0; \\
    \iwhile{t < T}{} \\
    \quad i \irnd [n]; \\
    \quad g \iass \nabla \ell(S[i],-) (w); \\
    \quad w \iass w - \alpha_t \cdot g; \\
    \quad t \iass t + 1; \\
    \iret{w}
  \end{array}
\]}
\codesgm
The program first initializes the parameters to some default value $w_0$. Then,
it runs $T$ iterations of the main
loop. The first step in the loop samples a uniformly random index $i$
from $[n] = \{ 0, 1, \dots, n - 1 \}$, while the second step computes the
gradient $g$. We will model the gradient operator $\nabla$ as a higher-order
function with type $(\RR^d \to [0, 1]) \to (\RR^d \to \RR^d)$.\footnote{%
  Strictly speaking, this operation is only well-defined if the input function
  is differentiable; this holds for many loss functions considered in the
machine learning literature.}
The third step in the loop updates $w$ to try to decrease the loss. The step
size $\alpha_t$ determines how far the algorithm moves; it is a real number that
may depend on the iteration $t$.

Our goal is to verify that this program is $\epsilon$-uniformly
stable. At a high level, suppose we have two training sets $S\lside$
and $S\rside$ differing in a single example; we write $\Adj(S\lside, S\rside)$.
Viewing the sets as
lists, we suppose that the two lists have the same length and
$S[i]\lside = S[i]\rside$ for all indices $i$ except for a one
index $j \in [n]$. We then construct an expectation coupling between
the two distributions on output parameters, bounding the expected
distance between the outputs $w\lside$ and $w\rside$. Assuming that
$\ell(z,-)$ is a Lipschitz function, i.e., $|\ell(z, w) - \ell(z,
w')| \leq L \| w - w' \|$ for all $w, w' \in \RR^d$ where $\| \cdot
\|$ is the usual Euclidean distance, bounding the expected distance
between the parameters also bounds the expected losses, implying
uniform stability.

Now, let's see how to carry out this verification in our logic. \SYSTEM is a
relational program logic with judgments of the form
\[
  {\Peprhl {s_1} {s_2} {\pre; \preexp} {\post; \postexp} {f}} .
\]
Here, $s_1, s_2$ are two imperative programs, the formulas $\pre$ and
$\post$ are assertions over pairs of memories $(m_1, m_2) \in \mem
\times \mem$, the distances $\preexp, \postexp$ are maps $\mem \times
\mem \to \RR^+$, and $f : \RR^+ \to \RR^+$ is a non-negative affine
function (i.e., of the form $x \mapsto ax + b$ for $a, b \in
\RR^+$). The judgment above states that for any two initial memories
$(m_1, m_2)$ satisfying the pre-condition $\pre$, there is an
expectation coupling $\mu$ of the output distributions from executing
$s_1, s_2$ on $m_1, m_2$ respectively such that the expected value of
$\postexp$ on the coupling is at most $f(\preexp(m_1, m_2))$ and all
pairs of output memories in the support of $\mu$ satisfy $\post$.

We focus on the loop. Let $s_a$ be the sampling
command and let $s_b$ be the remainder of the loop body.  First, we can show
\begin{equation} \label{eq:sgm-sample}
  \Peprhl{s_a}{s_a}
  {\Phi ; \| w\lside - w\rside \| }
  {i\lside = i\rside ; \| w\lside - w\rside \| }
  {\id} .
\end{equation}
The pre-condition $\Phi$ is shorthand for simpler invariants, including $t\lside
= t\rside$ and $\Adj(S\lside, S\rside)$. The post-condition $i\lside = i\rside$
indicates that the coupling assumes both executions sample the same index $i$.
Finally, the pre- and post-distances indicate that the expected value of $\|
w\lside - w\rside \|$ does not grow---this is clear because $s_a$ does not
modify $w$.

Now, we know that the training sets $S\lside$ and $S\rside$ differ in a single
example, say at index $j$. There are two cases: either we have sampled $i\lside
= i\rside = j$, or we have sampled $i\lside = i\rside \neq j$. In the first
case, we can apply properties of the loss function $\ell$ and gradient operator
$\nabla$ to prove:
\begin{equation} \label{eq:sgm-neq}
  \Peprhl{s_b}{s_b}
  {\Phi \land S[i] \lside \neq S[i]\rside ; \| w\lside - w\rside \| }
  { \Phi ; \| w\lside - w\rside \| }
  {\trans{\gamma}}
\end{equation}
where $\trans{\gamma}$ is the function $x \mapsto x + \gamma$ for some
constant $\gamma$ that depends on $L$ and the $\alpha_t$'s---since the algorithm
selects different examples in the two executions, the
resulting parameters may grow a bit farther apart.  In the second case,
the chosen example $S[i]$ is the same in both executions so we can prove:
\begin{equation} \label{eq:sgm-eq}
  \Peprhl{s_b}{s_b}
  {\Phi \land S[i] \lside = S[i]\rside ; \| w\lside - w\rside \| }
  { \Phi ; \| w\lside - w\rside \| }
  {\id} .
\end{equation}
The identity distance transformer $\id$ indicates that the expected distance
does not increase. To combine these two cases, we
note that the first case happens with probability $1/n$---this is the
probability of sampling index $j$---while the second case happens with
probability $1 - 1/n$. Our logic allows us to blend the distance transformers
when composing $s_a$ and $s_b$, yielding
\begin{equation} \label{eq:sgm-pcase}
  \Peprhl{s_a; s_b}{s_a; s_b}
  { \Phi ; \| w\lside - w\rside \| }
  { \Phi ; \| w\lside - w\rside \| }
  {+\gamma/n} ,
\end{equation}
since $x \mapsto (1/n) \cdot (x + \gamma) + (1 - 1/n) \cdot \id(x) = x + \gamma/n$.

Now that we have a bound on how the distance grows in the body, we can apply the
loop rule. Roughly speaking, for a loop running $T$ iterations, this rule simply
takes the $T$-fold composition $f^T$ of the bounding function $f$; since $f$ is
the linear function $+\gamma/n$, $f^T$ is the linear function  $+T\gamma/n$,
and we have:
\begin{equation} \label{eq:sgm-loop}
  \Peprhl{\sgm}{\sgm}
  { \Phi ; \| w\lside - w\rside \| }
  { \Phi ; \| w\lside - w\rside \| }
  {+T\gamma/n} .
\end{equation}
Assuming that the loss function $\ell(-,z)$ is Lipschitz, $|\ell(w, z) -
\ell(w', z)| \leq L \| w - w' \|$ for some constant $L$ and so
\begin{equation} \label{eq:sgm-lipschitz}
  \Peprhl{\sgm}{\sgm}
  { \Phi ; \| w\lside - w\rside \| }
  { \Phi ; |\ell(w\lside, z) - \ell(w\rside, z)| }
  {+LT \gamma / n}
\end{equation}
for every example $z \in Z$. Since $w\lside$ and $w\rside$ are
initialized to the same value $w_0$, the initial pre-distance is zero so this
judgment gives a coupling $\mu$ of the output distributions such that
\[
  \E{\mu}{| \ell(w\lside,z) - \ell(w\rside,z) |}
  \leq \| w_0 - w_0 \| + LT \gamma / n = LT \gamma / n .
\]
Since the left side is larger than
\[
  \E{\mu}{\ell(w\lside,z) - \ell(w\rside,z)}
  = \E{\mu_1}{\ell(w, z)} - \E{\mu_2}{\ell(w, z)} ,
\]
where $\mu_1$ and $\mu_2$ are the output distributions of $\sgm$, SGM is
$LT\gamma/n$-uniform stable.

\section{Expected Sensitivity} \label{sec:prelim}

Before we present our logic, we first review basic definitions and notations
from probability theory related to expected values and probabilistic couplings.
Then, we introduce our notions of expected sensitivity and expectation coupling.

\subsection{Mathematical preliminaries}

\paragraph*{Linear and Affine Functions.}
We let $\mathcal{A}$ be the set of non-negative affine functions, mapping
$z\mapsto \alpha \cdot z+ \beta$ where $\alpha, \beta\in\RR^+$;
$\mathcal{L} \subseteq \mathcal{A}$ be the set of non-negative linear functions,
mapping $z \mapsto \alpha \cdot z$;
$\mathcal{L}^\geq \subseteq \mathcal{L}$ be the set of non-contractive linear
functions, mapping $z \mapsto \alpha \cdot z$ with $\alpha \geq 1$; and
$\mathcal{C} \subseteq \mathcal{A}$ be the set of non-negative constant
functions, mapping $z \mapsto \beta$.
We will use the metavariables $f$ for $\mathcal{A}$ and bolded letters (e.g.,
$\const{\beta}$) for $\mathcal{C}$.

Non-negative affine functions can be combined in several ways. Let
$f,f'\in \mathcal{A}$ map $z$ to $\alpha \cdot z + \beta$ and $\alpha'
\cdot z + \beta'$ respectively, and let $\gamma\in\RR^+$.
\begin{itemize}
\item \emph{sequential composition:} the function $f'\circ f$ maps $z$
  to $(\alpha \alpha') \cdot z + (\alpha' \beta + \beta')$;
%\item \emph{iteration:} the function $f^k$ maps $z$ to
%  $\overbrace{f \circ \cdots \circ f}^{k~\mathrm{times}}(z)$
\item \emph{addition:} the function $f+f'$ maps $z$ to $f(z)+f'(z)$;
\item \emph{scaling:} the function $(\gamma \cdot f)$ maps $z$ to
  $\gamma \cdot f(z)$
\item \emph{translation:} the function $f \trans{\gamma}$ maps $z$ to
  $f(z) + \gamma$.
\end{itemize}
We will use shorthand for particularly common functions. For scaling, we write
$\scale{\gamma}$ for the function mapping $z$ to $\gamma \cdot z$. For
translation, we write $\trans{\gamma}$ for the function mapping $z$ to $z +
\gamma$. The identity function will be simply $\id$ (equivalently, $\scale{1}$
or $\trans{0}$).

%Note that affine functions are closed under all operations above, and
%linear functions are closed under composition, addition and scaling.

\paragraph*{Distances.}
%% We let $\RR$ be the set of real numbers and $\RRbar=\RR\cup \{ -
%% \infty, +\infty \}$ be the set of extended real numbers.  Moreover, we
%% let $\RR^+$ and $\RRbarpos$ denote the set of non-negative real numbers
%% and non-negative extended real numbers respectively. Throughout the paper,
%% we adopt the following convention: if $a>0$ and $b\geq 0$ then $a
%% \cdot (+\infty) + b = +\infty$.

A \emph{distance} $\preexp$ is a map $A \times A \to \RR^+$.  We use the term
``distance'' rather loosely---we do not assume any axioms, like reflexivity,
symmetry, triangle inequality, etc. Distances are partially ordered by the
pointwise order inherited from the reals: we write $\preexp \leq \postexp$ if
$\preexp(a_1, a_2) \leq \postexp(a_1, a_2)$ for all $(a_1, a_2) \in A \times A$.

\paragraph*{Distributions.}

Programs in our language will be interpreted in terms of sub-distributions.  A
(discrete) \emph{sub-distribution} over a set $A$ is a map $\mu : A \to \RR^+$
such that its \emph{support}
\[
  \supp(\mu) \eqdef \{ a \in A \mid \mu(a) \neq 0 \}
\]
is discrete and its \emph{weight} $\wt{\mu}\eqdef \sum_{a \in
  \supp(\mu) } \mu(a)$ is well-defined and satisfies $\wt{\mu} \leq
1$.  We let $\distr(A)$ denote the set of discrete sub-distributions
over $A$. Note that $\distr(A)$ is partially ordered using the
pointwise inequality inherited from reals. Similarly, equality of
distributions is defined extensionally: two distributions are equal if
they assign the same value (i.e., \emph{probability}) to each element
in their domain.  \emph{Events} are maps $E:A\to\BB$, where $\BB$
denotes the set of booleans. The probability of an event $E$ w.r.t. a
sub-distribution $\mu$, written $\Pr_\mu [E]$, is defined as
$\sum_{x \mid E(x)} \mu(x)$.

The \emph{expectation} of a function $f : A \to \RR^+$ w.r.t.\ a
sub-distribution $\mu \in \distr(A)$, written $\E {x \sim \mu} {f(x)}$
or $\E \mu f$ for short, is defined as $\sum_{x} \mu(x) \cdot f(x)$
when this sum exists, and $+\infty$ otherwise. Expectations are linear: $\E \mu
{f+g}=\E \mu f+\E \mu g$ and $\E \mu {k\cdot f}=k\cdot \E \mu f$, where addition
and scaling of functions are defined in the usual way.

Discrete sub-distributions support several useful constructions. First, they can
be given a monadic structure.  Let $x \in A$, $\mu \in \distr(A)$ and $M : A \to
\distr(B)$. Then:
\[
  \begin{array}{rcl}
    \mathsf{unit}~x & \triangleq & a \mapsto \ind{x = a} \\
    \mathsf{bind}~\mu~M & \triangleq & b \mapsto \sum_{a \in A} \mu(a) \cdot M(a)(b) .
  \end{array}
\]
Intuitively, $\mathsf{bind}~\mu~M$ is the distribution from first sampling from
$\mu$ and applying $M$ to the sample; in particular, it is a distribution over
$B$. We will write $\delta_x$ for the Dirac distribution $\mathsf{unit}~x$, and
abusing notation, $\E {x \sim \mu} {M}$ and $\E {\mu} {M}$ for
$\mathsf{bind}~\mu~M$. 

Given a distribution $\mu$ over pairs in $A \times B$, we can define the usual
projections $\pi_1 : \distr(A \times B) \to \distr(A)$ and $\pi_2 : \distr(A
\times B) \to \distr(B)$ as
\[
  \pi_1(\mu)(a) \triangleq \sum_{b \in B} \mu(a, b)
  \quad \text{and} \quad
  \pi_2(\mu)(b) \triangleq \sum_{a \in A} \mu(a, b) .
\]
A \emph{probabilistic coupling} is a joint distribution over pairs, such that
its first and second marginals coincide with the first
and second distributions. Formally, two sub-distributions $\mu_a\in\distr(A)$
and $\mu_b\in\distr(B)$ are \emph{coupled} by $\mu\in\distr(A\times B)$, written
$\coupling \mu {\mu_a} {\mu_b}$, if $\pi_1(\mu)=\mu_a$ and $\pi_2(\mu)=\mu_b$.

\subsection{Expected \texorpdfstring{$f$-sensitivity}{f-sensitivity}}
The core concept in our system is a probabilistic version of sensitivity.  Let
$f\in\mathcal{A}$, and let $\preexp_{A}$ and $\preexp_{B}$ be distances on $A$
and $B$.
\begin{definition}
  A probabilistic function $g:A\to\distr(B)$ is
  \emph{expected} $f$-\emph{sensitive} (with respect to $\preexp_A$ and
  $\preexp_B$) if for every $x_1,x_2\in A$, we have the bound
  \[
    \E{(y_1, y_2) \sim \mu}{\preexp_{B}(y_1, y_2)}\leq f(\preexp_{A}(x_1,x_2))
  \]
  for some coupling $\coupling {\mu} {g(x_1)} {g(x_2)}$.  When $f$ maps $z
  \mapsto \alpha \cdot z + \beta$, we sometimes say that $g$ is \emph{expected}
  $(\alpha, \beta)$-\emph{sensitive}.
\end{definition}
By carefully selecting the distances $\preexp_A$ and $\preexp_B$ on the input
and output spaces, we can recover different notions of probabilistic sensitivity
as a consequence of expected $f$-sensitivity. We derive two particularly useful
results here, which we first saw in the introduction.

\begin{proposition}[Average sensitivity] \label{prop:conseq:ex}
  Suppose that $g : A \to \distr(\RR)$ is expected $f$-sensitive with respect to
  distances $\preexp_A$ and $| \cdot |$. Then for any two inputs $a_1, a_2 \in
  A$, we have
  \[
    \left| \E {y_1 \sim g(a_1)} {y_1} - \E {y_2 \sim g(a_2)} {y_2} \right|
    \leq f(\preexp_A(a_1, a_2)) .
  \]
\end{proposition}
\begin{proof}
  Let $a_1, a_2 \in A$ be two inputs. Since $g$ is expected $f$-sensitive, there
  exists a coupling $\coupling {\mu} {g(a_1)} {g(a_2)}$ such that the expected
  distance over $\mu$ is at most $f(\preexp_A(a_1, a_2))$. We can bound:
  \begin{align}
    \left| \E {y_1 \sim g(a_1)} {y_1} - \E {y_2 \sim g(a_2)} {y_2} \right|
    &= \left| \E {(y_1, y_2) \sim \mu} {y_1} - \E {(y_1, y_2) \sim \mu} {y_2} \right|
    \tag{Coupling} \\
    &= \left| \E {(y_1, y_2) \sim \mu} {y_1 - y_2} \right|
    \tag{Linearity} \\
    &\leq \E {(y_1, y_2) \sim \mu} {\left| y_1 - y_2 \right|}
    \tag{Triangle ineq.} \\
    &\leq f(\preexp_A(a_1, a_2))
    \tag{$f$-sensitivity}
  \end{align}
\end{proof}

\begin{proposition}[Probabilistic sensitivity] \label{prop:conseq:tv}
  Suppose that $g : A \to \distr(B)$ is expected $f$-sensitive with respect to
  distances $\preexp_A$ and $\preexp_B$, where $\preexp_B(b_1, b_2) < \beta$
  if and only if $b_1 = b_2$.  Then for any two inputs $a_1, a_2 \in A$, we have
  \[
    \text{TV}(g(a_1), g(a_2)) \leq f(\preexp_A(a_1, a_2)) / \beta ,
  \]
  where the \emph{total variation distance} is defined as
  \[
    \text{TV}(g(a_1), g(a_2))
    \triangleq \max_{E \subseteq B}
    \left| \Pr_{b_1 \sim g(a_1)}[ b_1 \in E ] - \Pr_{b_2 \sim g(a_2)}[ b_2 \in E ] \right| .
  \]
\end{proposition}
\begin{proof}
  Let $a_1, a_2 \in A$ be two inputs. Since $g$ is expected $f$-sensitive, there
  exists a coupling $\coupling {\mu} {g(a_1)} {g(a_2)}$ such that the expected
  distance $\preexp_B$ over $\mu$ is at most $f(\preexp_A(a_1, a_2))$. We can
  bound:
  \begin{align}
    \Pr_{(b_1, b_2) \sim \mu} [ b_1 \neq b_2 ]
    &= \E {(b_1, b_2) \sim \mu} { \ind{b_1 \neq b_2} }
    \notag
    \\
    &\leq \E {(b_1, b_2) \sim \mu} { \preexp_B(b_1, b_2) / \beta }
    \tag{$b_1 \neq b_2 \to \preexp_B \geq \beta$}
    \\
    &\leq f(\preexp_A(a_1, a_2)) / \beta .
    \tag{Linearity, $f$-sensitivity}
  \end{align}
  By a classical theorem about couplings (see, e.g., \citet{Lindvall02}), the
  total-variation distance is at most the probability on the first line.
\end{proof}

Expected $f$-sensitive functions are closed under sequential composition.

\begin{proposition} \label{prop:f-lip-comp}
Let $f,f'\in\mathcal{A}$ be non-negative affine functions, and let
$\preexp_{A}$, $\preexp_{B}$ and $\preexp_C$ be distances on $A$, $B$
and $C$ respectively. Assume that $g:A\to\distr(B)$ is expected
$f$-sensitive and $h:B\to\distr(C)$ is expected $f'$-sensitive. Then
the (monadic) composition of $g$ and $h$ is expected $f' \circ f$-sensitive.
\end{proposition}
\begin{proof}
  Let $a_1, a_2 \in A$ be any pair of inputs. Since $g$
  is expected $f$-sensitive, there is a coupling $\coupling {\mu} {g(a_1)}
  {g(a_2)}$ such that
  \begin{equation} \label{eq:f-lip-comp-d1}
    \E {\mu} {\preexp_B} \leq f(\preexp_A(a_1, a_2)) .
  \end{equation}
  Similarly for
  every $b_1, b_2 \in B$, there is a coupling $\coupling {M(b_1, b_2)} {h(b_1)}
  {h(b_2)}$ such that
  \begin{equation} \label{eq:f-lip-comp-d2}
    \E {M(b_1, b_2)} {\preexp_C} \leq f'(\preexp_B(b_1, b_2)) ,
  \end{equation}
  since $h$ is $f'$-sensitive.

  Define the distribution $\mu' \triangleq \E {\mu} {M}$. It is
  straightforward to check the marginals $\pi_1(\mu')(a_1) = \E {g(a_1)} {h}$
  and $\pi_2(\mu')(a_2) = \E {g(a_2)} {h}$. We can bound the expected distance:
  \begin{align}
    \E {\mu'} {\preexp_C}
    &= \sum_{c_1, c_2} \preexp_C(c_1, c_2) \cdot \sum_{b_1, b_2} \mu(b_1, b_2) \cdot M(b_1, b_2)(c_1, c_2)
    \notag \\
    &= \sum_{b_1, b_2} \mu(b_1, b_2) \sum_{c_1, c_2} \preexp_C(c_1, c_2) \cdot M(b_1, b_2)(c_1, c_2)
    \notag \\
    &\leq \sum_{b_1, b_2} \mu(b_1, b_2) f'(\preexp_b(b_1, b_2))
    \tag{\cref{eq:f-lip-comp-d2}} \\
    &\leq f'\left( \sum_{b_1, b_2} \mu(b_1, b_2) \cdot \preexp_b(b_1, b_2)) \right)
    \tag{Linearity, $f'$ affine} \\
    &\leq f'\left( f(\preexp_A(a_1, a_2)) \right)
    \tag{\cref{eq:f-lip-comp-d1}, $f'$ non-decreasing} \\
    &= f' \circ f(\preexp_A(a_1, a_2))
    \notag .
  \end{align}
\end{proof}
Taking the pre- and post-distances to be the same yields another useful consequence.
\begin{proposition}
Let $\preexp$ be a distance over $A$ and let $f\in\mathcal{A}$. Let
$g:A\to \distr(A)$ be an expected $f$-sensitive function. Then for
every $T\in\mathbb{N}$, the $T$-fold (monadic) composition $g^T$ of $g$
is expected $f^T$-sensitive, i.e.\, for every $x_1,x_2\in A$, there
exists a coupling $\coupling {\mu} {g^{T}(x_1)} {g^{T}(x_2)}$ such that
\[
  \E \mu {\preexp} \leq f^T(\preexp(x_1,x_2)) .
\]
\end{proposition}
This proposition can be seen as a variant of the Banach
fixed point theorem. Informally, under some
reasonable conditions on $\preexp$, contractive probabilistic maps
$g:A\to\distr(A)$ have a unique stationary distribution, where
a probabilistic map is \emph{contractive} if it is expected $f$-sensitive for a
map $f$ of the form $z\mapsto\alpha\cdot z$ with $\alpha<1$.

\subsection{Continuity from Expectation Couplings}
Expected $f$-sensitivity is a property of a probabilistic function. It will be
useful to factor out the condition on distributions.  To this end, we introduce
\emph{expectation couplings} a quantitative extension of probabilistic
couplings where an average distance over the coupling is bounded.
\begin{definition}[Expectation couplings]
  Let $\preexp: A \times B \to \RR^+$ be a distance and let
  $\delta \in \RR^+$ be a constant. Moreover, let $\mu_a \in
  \distr(A)$, $\mu_b \in \distr(B)$ and $\mu \in \distr(A \times
  B)$. Then $\mu$ is an $(\preexp, \delta)$-\emph{expectation
    coupling} (or simply, an \emph{expectation coupling}) for $\mu_a$
  and $\mu_b$ if $\coupling \mu {\mu_a} {\mu_b}$ and $\E \mu \preexp
  \le \delta$.

We write $\bcouplingsupp {\mu} {\preexp} {\delta} {\mu_a} {\mu_b}
{\pre}$ when $\mu$ is an expectation coupling with support
$\supp(\mu)$ contained in a binary relation $\pre \subseteq A \times
B$. We omit $\pre$ when it is the trivial (always true) relation.
\end{definition}       
Expectation couplings are closely linked to expected $f$-sensitivity.
\begin{proposition}
A probabilistic function $g:A\to\distr(B)$ is \emph{expected}
$f$-\emph{sensitive} (with respect to $\preexp_A$ and $\preexp_B$) if for every
$x_1,x_2\in A$, there exists $\mu$ such that $\bcouplingsupp {\mu} {\preexp}
{\delta} {g(x_1)} {g(x_2)}{}$, where $\delta=f(\preexp_{A}(x_1,x_2))$.
\end{proposition}
Much like expected $f$-sensitive functions, 
expectation couplings are closed under sequential composition: given
an expectation coupling between two distributions $\mu_a$ and $\mu_b$,
two functions $M_a : A \to \distr(A')$ and $M_b : B \to \distr(B')$
and a function $M$ mapping pairs of samples in $(a, b) \in A \times B$
to an expectation coupling of $M_a(a)$ and $M_b(b)$, we have an
expectation coupling of the two distributions from sampling $\mu_a$
and $\mu_b$ and running $M_a$ and $M_b$, respectively.
\begin{proposition}[Composition of expectation couplings]
  \label{prop:seqcomp}
  Let $\pre \subseteq A\times B$, $\preexp: A \times B \to \RR^+$,
  $\post \subseteq A\times B$, $\postexp: A \times B \to \RR^+$,
  $\delta \in \RR^+$, and $f\in\mathcal{A}$. Let
  $\mu_a\in\distr(A)$, $M_a:A\to\distr(A')$, and let $\mu'_a =
  \E {\mu_a} {M_a}$. Let $\mu_b\in\distr(B)$,
  $M_b:B\to\distr(B')$, and set $\mu'_b = \E {\mu_b} {M_b}$. Suppose we
  have functions $\mu\in\distr(A\times B)$ and $M:(A\times B) \to
  \distr(A'\times B')$ such that:
  \begin{enumerate}
  \item $\bcouplingsupp \mu \preexp \delta {\mu_a} {\mu_b} {\pre}$ and
  \item $\bcouplingsupp {M(a,b)} \postexp {f(\preexp(a,b))}  {M_a(a)}
    {M_b(b)} {\post}$ for every $(a,b) \in \pre$.
  \end{enumerate}
  Then $\bcouplingsupp {\mu'} \postexp {f(\delta)} {\mu'_a} {\mu'_b} {\post}$,
  where $\mu'$ is the monadic composition $\E {(a, b) \sim \mu} {M(a, b)}$.
\end{proposition}
\begin{proof}[Proof sketch]
  By unfolding definitions and checking the support, marginal, and expected
  distance properties. The support and marginal conditions follow by the support
  and marginal conditions for the premises, while the expected distance
  condition follows by an argument similar to \cref{prop:f-lip-comp}. We defer
  the details to the \LONGVERSION.
\end{proof}

\section{Program Logic} \label{sec:logic}

As we have seen, expectation couplings can be composed together and the
existence of an expectation coupling implies expected sensitivity. Accordingly,
we can give a program logic to reason about expectation couplings in a
structured way.

\subsection{Programming Language}

We base our development on \pwhile, a core language with deterministic
assignments, probabilistic assignments, conditionals, and loops. The syntax of
statements is defined by the grammar:
\begin{align*}
    s &::= x \iass e
           \mid x \irnd g 
           \mid s; s \mid \iskip
        \mid \ifte{e}{s}{s}
        \mid \iwhile{e}{s}
\end{align*}
where $x$, $e$, and $g$ range over variables $\vars$, expressions
$\exprs$ and distribution expressions $\dexprs$ respectively.
$\exprs$ is defined inductively from $\vars$ and operators, while
$\dexprs$ consists of parametrized distributions---for instance, the
uniform distribution ${[n]}$ over the set $\{ 0,\ldots, n-1\}$ or the
Bernoulli distribution $\bernD(p)$, where the numeric parameter $p \in [0, 1]$
is the probability of returning true. We
will write $\ift{e}{s}$ as shorthand for $\ifte{e}{s}{\iskip}$.
We implicitly assume that programs are well-typed w.r.t.\, a standard typing
discipline; for instance, the guard expressions of conditionals and loops are
booleans, operations on expressions are applied to arguments of the correct
type, etc.

Following the seminal work of \citet{Kozen79}, probabilistic programs can be
given a monadic denotational semantics, taking a memory as input and producing a
sub-distribution on output memories. To avoid measure-theoretic technicalities,
we limit our focus to discrete sub-distributions.
Memories are type-preserving maps from variables to values---formally, we
define an interpretation for each type and require that a variable of type $T$
is mapped to an element of the interpretation of $T$. We let $\mem$ denote the
set of memories.  Then, the semantics $\dsem{m}{e}$ of a (well-typed) expression
$e$ is defined in the usual way as an element of the interpretation of the type
of $e$, and parametrized by a memory $m$. The interpretation of distribution
expressions is defined and denoted likewise.

Now, the semantics $\dsem{m}{s}$ of a statement $s$ w.r.t.\ to some initial
memory $m$ is the sub-distribution over states defined by the clauses of
\cref{fig:semantics}.
The most interesting case is for loops, where the interpretation of a
$\kwhile$ loop is the limit of the interpretations of its finite unrollings.
Formally, the $n^{th}$ \emph{truncated iterate} of the loop
$\iwhile{b}{s}$ is defined as
$$\overbrace{\ift{b}{s}; \ldots; \ift{b}{s}}^{n~\mbox{times}};\ift{b}{\iabort}$$
which we represent using the shorthand $(\ift{b}{s})^n_{\mid \neg b}$.  For any
initial memory $m$, applying the truncated iterates yields an
pointwise-increasing and bounded sequence of sub-distributions. The limit of
this sequence is well-defined, and gives the semantics of the $\kwhile$ loop.

\begin{figure*}
\begin{align*}
  \dsem{m}{\iskip} &= \dunit{m} &
  \dsem{m}{x \irnd g} &= \dlet v {\dsem{m}{g}} {\dunit{m[\subst{x}{v}]}} \\
  \dsem{m}{x \iass e} &= \dunit{m[\subst{x}{\dsem{m}{e}}]} &
  \dsem{m}{\ifte{e}{s_1}{s_2}} &=
    \text{if $\dsem{m}{e}$ then $\dsem{m}{s_1}$ else $\dsem{m}{s_2}$} \\
  \dsem{m}{s_1; s_2} &= \dlet {\xi} {\dsem{m}{s_1}} {\dsem{\xi}{s_2}} &
  \dsem{m}{\iwhile{b}{s}} &=
    \lim_{n \to \infty}\ \dsem{m}{(\ift b s)^n_{\mid \neg b}}
\end{align*}
(Note that $\mathbb{E}$ is the monadic bind.)
\caption{\label{fig:semantics} Denotational semantics of programs}
\end{figure*}

\subsection{Proof System}
\SYSTEM is a Hoare-style logic augmented to consider two programs instead of
one (a so-called \emph{relational} program logic). \SYSTEM judgments are of the form
\[\eprhl{s_1}{s_2}{\pre; \preexp}{\post; \postexp}{f}\]
for programs $s_1$, $s_2$, assertions $\pre,\post: \mem \times \mem
\to \BB$, distances $\preexp, \postexp : \mem \times \mem \to
\RR^+$, and a non-negative affine function $f\in\mathcal{A}$. We will refer
to $f$ as a \emph{distance transformer}.
\begin{definition} 
A judgment $\eprhl{s_1}{s_2}{\pre; \preexp}{\post; \postexp}{f}$
is valid if for
every memories $m_1$, $m_2$ s.t. $(m_1, m_2) \models \pre$, there
exists $\mu$ such that
$$\bcouplingsupp{\mu}{\postexp}{f(\preexp(m_1, m_2))}{\dsem{m_1}{s_1}}{
  \dsem{m_2}{s_2}}{\post}$$
\end{definition}

The notion of validity is closely tied to expected $f$-sensitivity. For
instance, if the judgment
\[
  \eprhl{s}{s}{\top; \preexp}{\top; \postexp}{f}
\]
is valid, then the program $s$ interpreted as a function $\dsem{}{s} : \mem \to
\distr(\mem)$ is expected $f$-sensitive with respect to distances $\preexp$ and
$\postexp$. In fact, the pre- and post-conditions $\pre$ and $\post$ can also be
interpreted as distances. If we map $\pre$ to the pre-distance
$\preexp_\pre(m_1, m_2) \triangleq \ind{ (m_1, m_2) \notin \pre }$, and
$\post$ to the post-distance $\preexp_\post(m_1, m_2) \triangleq \ind{
(m_1, m_2) \notin \post }$, then the judgment
\[
  \eprhl{s_1}{s_2}{\top; \preexp_\pre}{\top; \preexp_\post}{\id}
\]
is equivalent to
\[
  \eprhl{s_1}{s_2}{\pre; -}{\post; -}{-}
\]
where dashes stand for arbitrary distances and distance transformers.

Now, we introduce some notation and then present the rules of the logic.
First, note that each boolean expression $e$ naturally yields two assertions
$e\lside$ and $e\rside$, resp. called its left and right injections:
\[\begin{aligned}
  m_1 \models e &\iff m_1, m_2 \models e\lside \\
  m_2 \models e&\iff m_1, m_2 \models e\rside 
\end{aligned} \]
The notation naturally extends to mappings from memories to booleans.
Second, several rules use substitutions. Given a memory $m$, variable
$x$ and expression $e$ such that the types of $x$ and $e$ agree, we
let $m[\subst{x}{e}]$ denote the unique memory $m'$ such that $m(y) =
m'(y)$ if $y\neq x$ and $m'(x)=\dsem{m}{e}$. Then, given a variable
$x$ (resp. $x'$), an expression $e$ (resp. $e'$), and an assertion
$\pre$, we define the assertion
$\pre[\subst{x\lside,x'\rside}{e\lside,e'\rside}]$ by the clause
\[
  \pre[\subst{x\lside,x'\rside}{e\lside,e'\rside}](m_1,m_2)
    \eqdef \pre(m_1[\subst{x}{e}],m_2[\subst{x'}{e'}]) .
\]
Substitution of distances is defined similarly. One can also define
one-sided substitutions, for instance $\pre[\subst{x\lside}{e\lside}]$.

We now turn to the rules of the proof system in
\cref{fig:rules}. The rules can be divided into two groups: two-sided
rules relate programs with the same structure, while structural rules apply to
two programs of any shape. The full logic \SYSTEM also features one-sided rules
for relating a program with a fixed shape to a program of unknown shape; later
we will show that many of these rules are derivable. We briefly comment on each
of the rules, starting with the two-sided rules.

The [\textsc{Assg}] rule is similar to the usual rule for assignments, and
substitutes the assigned expressions into the pre-condition and pre-distance.

The [\textsc{Rand}] rule is a bit more subtle. Informally, the rule selects
a coupling, given as a bijection between supports, between the
two sampled distributions in the left and right program.

The [\textsc{SeqCase}] rule combines sequential composition with a case analysis
on properties satisfied by intermediate memories after executing
$s_1$ and $s_2$. Informally, the rule considers events $e_1, \ldots, e_n$
such that $\post$ entails $\bigvee_i {e_i}\lside$. If for every $i$ we can
relate the programs $s'_1$ and $s'_2$ with distance transformer $f_i$,
pre-condition $\post \land {e_i}\lside ;\postexp$ and post-condition
$\postz;\postexpz$, we can conclude that $s_1;s'_1$ and $s_2;s'_2$ are related
under distance transformer $f$, where $f$ upper bounds the functions $f_i$
weighted by the probability of each case.

The [\textsc{While}] rule considers two loops
synchronously, where the loop bodies preserve the invariant
$\post$. The rule additionally requires that both loops
perform exactly $n$ steps, and that there exists a variant $i$
initially set to $n$ and decreasing by 1 at each iteration. Assuming
that $f_k$ denotes the distance transformer corresponding to the
$(n-k)$th iteration, i.e., the iteration where the variant $i$ is equal
to $k$, the distance transformer for the \kwhile loops is the
function composition of the distance transformers: $f_1 \circ \dots \circ f_n$.

The remaining rules are structural rules.
The [\textsc{Conseq}] rule weakens the post-conditions, strengthens the
pre-conditions, and relaxes the distance bounds.

The [\textsc{Struct}] rule
replaces programs by equivalent programs. \Cref{fig:equiv} gives rules for
proving two programs $s, s'$ equivalent under some relational assertion $\pre$;
the judgments are of the form $\eqsem{\pre}{s}{s'}$. We keep the notion of
structural equivalence as simple as possible.

The [\textsc{Frame-D}] rule generalizes the
typical frame rule, to preserve distances. 
Assuming that the distance $\postexpz$ is not modified by
the statements of the judgments and $f$ is a non-contractive linear function
(i.e., such that $x\leq f(x)$ for all $x$), validity is preserved when adding
$\postexpz$ to the pre-and post-distances of the judgment.
Formally, $\MV(s)$ denotes the set of modified variables of $s$ and
the notation $\postexp' \# \MV(s_1),\MV(s_2)$ states that for all
memories $m_1$ and $m_1'$ that coincide on the non-modified variables
of $s_1$, and all memories $m_2$ and $m_2'$ that coincide on the non-modified
variables of $s_2$, we have $\postexp'(m_1,m_2)=\postexp'(m'_1,m'_2)$.

\begin{theorem}[Soundness] \label{thm:soundness}
For every derivable judgment
$\Peprhl{s_1}{s_2}{\pre;\preexp}{\post;\postexp}{f}$ and
initial memories $m_1$ and $m_2$ such that $(m_1, m_2) \models \pre$, there
exists $\mu$ such that
$$\bcouplingsupp{\mu}{\postexp}{f(\preexp(m_1, m_2))}{\dsem{m_1}{s_1}}{
  \dsem{m_2}{s_2}}{\post} .$$
\end{theorem}
\begin{proof}
  By induction on the derivation. We defer the details to the \LONGVERSION.
\end{proof}

\begin{figure*}
  \begin{mathpar}
    
%% \inferrule[Skip]
%% { }
%% {\Peprhl {\iskip} {\iskip} {\post; \postexp} {\post; \postexp} {\id}}

\inferrule[Assg]
  {~~}
  { \Peprhl
       {x_1 \iass e_1}{x_2 \iass e_2}
       {\post[\subst{{x_1}\lside}{e_1}, \subst{{x_2}\rside}{e_2}];
         \postexp[\subst{{x_1}\lside}{e_1}, \subst{{x_2}\rside}{e_2}]
       }{\post;\postexp}
       {\id} }

\inferrule[Rand]
          { % g_1\in\distr(A_1) \\ g_2\in\distr(A_2) \\
            h: \supp(g_1) \tobij \supp(g_2) \\
            \preexp \eqdef \E {v \sim g_1}
      {\postexp[\subst{{x_1}\lside}{v},\subst{{x_2}\rside}{h(v)}]} \\
    \forall v \in \supp(g_1) .\, g_1(v) = g_2(h(v))}
  { \Peprhl {x_1 \irnd g_1}{x_2 \irnd g_2}
      {\forall v \in \supp(g_1) .\,
        \post [\subst{{x_1}\lside}{v}, \subst{{x_2}\rside}{h(v)}];
         \preexp}
      {\post;\postexp} {\id} }

\inferrule[SeqCase]
 { \forall m_1, m_2 \models \pre .\,
     (\textstyle\sum_{i \in I} {\textstyle\Pr}_{\dsem{m_1}{s_1}} [e_i] \cdot f_i)\circ f_0 \leq f \\
    \models \post \implies {\textstyle\bigvee}_{i \in I} {e_i}\lside
    \\\\
    \Peprhl{s_1}{s_2}{\pre ; \preexp}{\post ; \postexp}{f_0} \\
    \forall i \in I.\,
     \Peprhl {s'_1} {s'_2}
       {\post \land {e_i}\lside; \postexp}{\postz; \postexpz}{f_i} 
%    \forall m_1, m_2 \models \pre.\,
%      {\textstyle\bigwedge}_i
%        {\textstyle\Pr}_{\dsem{m_1}{s_1}} [e_i] \leq p_i
     }
 { \Peprhl{s_1; s'_1}{s_2; s'_2}
     {\pre ; \preexp}{\postz; \postexpz}
     {f} }

  \inferrule[While]
  { \models \post \implies e\lside = e\rside \land (i\lside \leq 0 \iff \neg e\lside)
    \\
    \forall 0 < k \leq n .\,
       \Peprhl {s_1}{s_2}
         {\post \land {e_1}\lside 
                \land i\lside =k ; \preexp_{k}}
              {\post\land i\lside = k-1 ; \preexp_{k - 1}}{f_k} }
  { \Peprhl {\iwhile{e_1}{s_1}} {\iwhile{e_2}{s_2}}
    {\post\land i\lside =n; \preexp_n} {\post \land i\lside =0; \preexp_0}
    { f_1\circ \cdots  \circ f_n} }

\inferrule[Conseq]
 { \Peprhl{s_1}{s_2}{\pre; \preexp}{\post; \postexp}{f} \\
   \models \pre' \implies \pre \\
   \models \post \implies \post' \\
   \models \pre' \implies f(\preexp) \leq f'(\preexp'') \\
   \models \post \implies \postexp'' \leq
   \postexp }
 { \Peprhl{s_1}{s_2}{\pre';\preexp''}{\post'; \postexp''}{f'} }

\inferrule[Struct]
          { \Peprhl{s_1}{s_2}{\pre;\preexp}{\post; \postexp}{f} \\
            \eqsem{\pre_1}{s_1}{s_1'} \\ \eqsem{\pre_2}{s_2}{s_2'} \\
    \forall (m_1,m_2) \models \pre .\, \pre_1(m_1) \land \pre_2(m_2) 
    }
  { \Peprhl{s_1'}{s_2'}{\pre; \preexp}{\post; \postexp}{f} }

  \inferrule[Frame-D]
  { \Peprhl{s_1}{s_2}{\pre ; \preexp}{\post ; \postexp}{f}
    \\\\
    f\in\mathcal{L}^\geq \\ \postexp' \# \MV(s_1),\MV(s_2) \\ \models \pre \implies \postexp' \leq f(\postexp') }
  { \Peprhl{s_1}{s_2}{\pre ; \preexp + \postexp'}{\post ; \postexp + \postexp' }{f} }
\end{mathpar}

\caption{\label{fig:rules} Selected proof rules}
\end{figure*}

\begin{figure*}

\begin{mathpar}
 \inferrule[Seq]
  { \Peprhl{s_1}{s_2}{\pre;\preexp}{\Xi; \postexp}{f} \\
    \Peprhl{s'_1}{s'_2}{\Xi; \postexp}{\post; \postexpz}{f'}}
  { \Peprhl{s_1;s'_1}{s_2;s'_2}{\pre; \preexp}{\post; \postexpz}{f' \circ f} }

  \inferrule[Case]
 { \Peprhl{s_1}{s_2}{\pre\land e\lside; \preexp}{\post; \postexp}{f} \\
   \Peprhl{s_1}{s_2}{\pre\land \neg e\lside; \preexp}{\post; \postexp}{f} }
 { \Peprhl{s_1}{s_2}{\pre; \preexp}{\post; \postexp}{f} }

 \inferrule[Cond]{\models \pre \implies {e_1}\lside={e_2}\rside \\
   \Peprhl{s_1}{s_2}{\pre\land {e_1}\lside; \preexp}{\post; \postexp}{f} \\
   \Peprhl{s'_1}{s'_2}{\pre\land \neg {e_1}\lside; \preexp}{\post; \postexp}{f} }
 { \Peprhl{\ifte{e_1}{s_1}{s'_1}}{\ifte{e_2}{s_2}{s'_2}}{\pre; \preexp}{\post; \postexp}{f}}

\inferrule[Assg-L]
  {~~}
  { \Peprhl
       {x_1 \iass e_1}{\iskip}
       {\post[\subst{{x_1}\lside}{e_1}];
         \postexp[\subst{{x_1}\lside}{e_1}]
       }{\post;\postexp}
       {\id} } 
\end{mathpar}

\caption{\label{fig:derived:rules} Selected derived rules}
\end{figure*}

\subsection{Derived Rules and Weakest Pre-condition}

\Cref{fig:derived:rules} presents some useful derived rules of our logic,
including rules for standard sequential composition and conditionals, and
one-sided rules.

The [\textsc{Seq}] rule for sequential composition simply composes the
two product programs in sequence. This rule reflects the compositional
property of couplings. It can be derived from the rule
[\textsc{SeqCase}] by taking $e_1$ to be true.

The [\textsc{Cond}] rule for conditional statements requires that the
two guards of the left and right programs are equivalent under the
pre-condition, and that corresponding branches can be related.

The [\textsc{Case}] rule allows proving a judgment by case analysis;
specifically, the validity of a judgment can be established from the
validity of two judgments, one where the boolean-valued pre-condition
is strengthened with $e$ and the other where the pre-condition is
strengthened with $\neg e$.

The [\textsc{Assg-L}] is the left one-sided rule for assignment; it
can be derived from the assignment rule using structural equivalence.
The full version of the logic also has similar one-sided rules for other constructs, notably
random assignments and conditionals. Using one sided-rules, one can also
define a relational weakest pre-condition calculus $\wpc$, taking as
inputs two loop-free and deterministic programs, a post-condition, and
a distance, and returning a pre-condition and a distance.  
\begin{proposition}
Let $(\pre'',\preexp'')=\wpc(s_1,s_2,\post,\postexp)$. Assume
$\pre\implies\pre''$ and $\preexp(m_1,m_2)\leq\preexp''(m_1,m_2)$ for
every $(m_1,m_2)\models\pre$. Then
$\Peprhl{s_1}{s_2}{\pre;\preexp}{\post; \postexp}{\id}$.
\end{proposition}
%% For practical applications it is convenient to extend the weakest
%% pre-condition calculus to loop-free programs by applying the random
%% sampling rule with the identity bijection---note that this rule is
%% incomplete with respect to the proof system, because the rule for
%% random assignments allows to use arbitrary bijections, but many
%% applications take the identity function as bijection.

%% \subsection{Embedding \Sprhl}
%% \gb{I would like to discuss embedding our logic into one with multiple
%%   distances, and no relational assertions}

%% The relational programming logic \Sprhl can be embedded into our
%% program logic, in the following sense.
%% \begin{proposition}
%% If $\Pprhl{s_1}{s_2}{\pre}{\post}$ then for every real-valued
%% assertion $\postexp$ and distance transformer $f\in\mathcal{A}$,
%% $\Peprhl{s_1}{s_2}{\pre;\lambda \_. +\infty}{\post;\postexp}{f}$
%% \end{proposition}

\begin{figure}
\begin{mathpar}
\inferrule{~}{\eqsem{\pre}{s}{s}}

\inferrule{\eqsem{\pre}{s_1}{s_2}}{\eqsem{\pre}{s_2}{s_1}}

\inferrule{~}{\eqsem{\pre}{{x}\irnd {\dunit{x}}}{\iskip}}

\inferrule{\models \pre \implies x = e}{\eqsem{\pre}{{x}\iass {e}}{\iskip}}

\inferrule{~}{\eqsem{\pre}{s;\iskip}{s}}

\inferrule{~}{\eqsem{\pre}{\iskip;s}{s}}

\inferrule{\eqsem{\pre}{s_1}{s'_1}}
          {\eqsem{\pre}{s_1;s_2}{s'_1;s_2}}

\inferrule{\eqsem{\top}{s_2}{s'_2}}
          {\eqsem{\pre}{s_1;s_2}{s_1;s'_2}}

\inferrule{\models \pre \implies e}{\eqsem{\pre}{\ifte{e}{s}{s'}}{s}}

\inferrule{\models \pre \implies \neg e}{\eqsem{\pre}{\ifte{e}{s}{s'}}{s'}}

\inferrule
  {\eqsem{\pre\land e}{s_1}{s_2} \\
   \eqsem{\pre\land \neg e}{s'_1}{s'_2}}
  {\eqsem{\pre}{\ifte{e}{s_1}{s'_1}}{\ifte{e}{s_2}{s'_2}}}

%% \inferrule{\eqsem{e}{s}{s'}}{\eqsem{\pre}{\iwhile{e}{s}}{\iwhile{e}{s'}}}

%% \inferrule{~}{\eqsem{\pre}{\iwhile{e}{s}}{\ift{e}{(s;\iwhile{e}{s})}}}

\end{mathpar}

\caption{\label{fig:equiv} Equivalence rules}
\end{figure}

\section{Uniform Stability of Stochastic Gradient Method, Revisited} \label{sec:sgm}

Now that we have described the logic, let's return to the Stochastic
Gradient Method we first saw in \cref{sec:motiv}. Recall that the loss
function has type $\ell : Z \to \RR^d \to [0, 1]$. We consider two
versions: one where the loss function $\ell(z,-)$ is convex, and one
where $\ell(z,-)$ may be non-convex. The algorithm is the same in both
cases, but the stability properties require different proofs. For
convenience, we reproduce the code $\sgm$:
\codesgm
We will assume that $\ell(z,-)$ is $L$-\emph{Lipschitz} for all $z$: for all $w,
w' \in \RR^d$, we can bound $|\ell(z, w) - \ell(z, w')| \leq L \| w - w' \|$
where $\|\cdot\|$ is the usual Euclidean norm on $\RR^d$:
\[
  \| x \| \triangleq \left( \sum_{i = 1}^d x_i^2 \right)^{1/2}
\]
Furthermore, we will assume that the loss function is $\beta$-\emph{smooth}: the
gradient $\nabla \ell(z, -) : \RR^d \to \RR^d$ must be $\beta$-Lipschitz.

\subsection{SGM with Convex Loss} \label{ex:convex-sgm}

Suppose that the loss $\ell(z,-)$ is a \emph{convex} function for every $z$,
i.e., we have: $\langle (\nabla \ell(z, -))(w) - (\nabla \ell(z, -))(w'), w - w'
\rangle \geq 0$ where $\langle x, y \rangle$ is the inner product between two
vectors $x, y \in \RR^d$:
\[
  \langle x, y \rangle \triangleq \sum_{i = 1}^d x_i \cdot y_i .  
\]
When the step sizes satisfy $0 \leq \alpha_t \leq 2/\beta$, we can prove uniform
stability of SGM in this case by following the strategy outlined in
\cref{sec:motiv}. We refer back to the judgments there, briefly describing how
to apply the rules (for lack of space, we defer some details to the
\LONGVERSION).  Let $s_a$ be the sampling command, and $s_b$ be the rest of the
loop body. We will prove the following judgment:
\[
  \Peprhl{\sgm}{\sgm}
  { \Phi ; \| w\lside - w\rside \| }
  { \Phi ; |\ell(w\lside, z) - \ell(w\rside, z)| }
  {+\gamma} ,
\]
where $\Phi \triangleq \Adj(S\lside, S\rside) \land (w_0)\lside = (w_0)\rside
\land t\lside = t\rside$ and
\[
  \gamma \triangleq \frac{2L^2}{n} \sum_{t = 0}^{T - 1} \alpha_t .
\]
By soundness (\cref{thm:soundness}), this will imply that SGM is
$\gamma$-uniformly stable.

As before, we will first establish a simpler judgment:
\[
  \Peprhl{\sgm}{\sgm}
  { \Phi ; \| w\lside - w\rside \| }
  { \Phi ; \| w\lside - w\rside \| }
  {+\gamma/L} .
\]
As we proceed through the proof, we will indicate the corresponding step from
the outline in \cref{sec:motiv}.
Let $j$ be the index such that the $S[j]\lside \neq S[j]\rside$; this is the
index of the differing example. First, we couple the samplings in $s_a$ with the
identity coupling, using rule \rname{Rand} with $h = \id$
(\cref{eq:sgm-sample}). Next, we perform a case analysis on whether we sample
the differing vertex or not. We can define guards $e_= \triangleq i = j$ and
$e_{\neq} \triangleq i \neq j$, and then apply the probabilistic case rule
\rname{SeqCase}. In the case $e_=$, we use the Lipschitz property of
$\ell(z,-)$ and some properties of the norm $\| \cdot \|$ to prove
\[
  \Peprhl{s_b}{s_b}
  { \Phi \land e_= ; \| w\lside - w\rside \| }
  { \Phi ; \| w\lside - w\rside \| }
  {+ 2\alpha_t L} ;
\]
this corresponds to \cref{eq:sgm-neq}. In the case $e_{\neq}$, we know that the
examples are the same in both runs. So, we use the Lipschitz property,
smoothness, and convexity of $\ell(z, -)$ to prove:
\[
  \Peprhl{s_b}{s_b}
  { \Phi \land e_{\neq} ; \| w\lside - w\rside \| }
  { \Phi ; \| w\lside - w\rside \| }
  {\id} ;
\]
this corresponds to \cref{eq:sgm-eq}. Applying \rname{SeqCase}, noting that the
probability of $e_{\neq}$ is $1 - 1/n$ and the probability of $e_=$ is $1/n$, we
can bound the expected distance for the loop body (\cref{eq:sgm-pcase}).
Applying the rule \rname{While}, we can bound the distance for the whole
loop (\cref{eq:sgm-loop}). Finally, we use the Lipschitz property of
$\ell(z, -)$ and the rule \rname{Conseq} to prove the desired judgment.

\subsection{SGM with Non-Convex Loss} \label{ex:nonconvex-sgm}

When the loss function is non-convex, the previous proof no longer
goes through.  However, we can still verify the uniform stability
bound by \citet{HardtRS16}.
Technically, they prove uniform stability by dividing the
proof into two pieces. First they show that with sufficiently high
probability, the algorithm does not select the differing example
before a carefully chosen time $t_0$. In particular, with high probability
the parameters $w\lside$ and $w\rside$ are equal up to iteration
$t_0$. Then, they prove a uniform stability bound for SGM started at
iteration $t_0$, assuming $w\lside = w\rside$; if the
step size $\alpha_t$ is taken to be rapidly decreasing, SGM will be already be
contracting by iteration $t_0$.

This proof can also be carried out in \SYSTEM, with some extensions. First, we
split the SGM program into two loops: iterations before $t_0$, and iterations
after $t_0$. The probability of $w\lside \neq w\rside$ is is precisely
the expected value of the indicator function $\ind{w\lside \neq w\rside}$, which
is $1$ if the parameters are not equal and $0$ otherwise. Thus, we can bound the
probability for the first loop by bounding this expected value in \SYSTEM.
For the second loop, we can proceed much like we did for standard SGM: assume
that the parameters are initially equal, and then bound the expected
distance on parameters.

The most difficult part is gluing these two pieces together.  Roughly, we want
to perform a case analysis on $w\lside = w\rside$ but this event depends on both
sides---the existing probabilistic case rule \rname{SeqCase} does not suffice.
However, we can give an advanced probabilistic case rule \rname{SeqCase-A} that
does the trick. We defer the details to the \LONGVERSION.

\section{Population Dynamics} \label{sec:population} 
Our second example comes from the field of evolutionary biology. Consider an
infinite population separated into $m\in\NN$ classes of organisms. The population
at time $t$ is described by a probability vector $\vec{x}_t = (x_1, \ldots,
x_m)$, where $x_i$ represents the fraction of the population belonging to
class $i$. In the Replication-Selection-Mutate (RSM) model, the evolution is
described by a function
$f$---called the \emph{step function}---which updates the probability vectors.
More precisely, the population at time $t+1$ is given as the average of $N \in
\NN$ samples according to the distribution $f(\vec{x}_t)$. A central question is
whether this process mixes rapidly: starting from two possibly different
population distributions, how fast do the populations converge?

We will verify a probabilistic property that is the main result needed to show
rapid mixing: there is a coupling of the population distributions such that the
expected distance between the two populations decreases exponentially quickly.
Concretely, we take the norm $\| \vec{x} \|_1 \triangleq \sum_{i = 1}^m |x_i|$.
Let the simplex $\Delta_m$ be the set of non-negative vectors with norm $1$:
\[
  \Delta_m \triangleq \{  \vec{x} \in \RR^m \mid x_i \geq 0,  \| \vec{x} \|_1 = 1 \}
\]
Elements of $\Delta_m$ can be viewed as probability distributions over
the classes $\{ 1, \dots, m \}$; this is how we will encode the
distribution of species.

In the RSM model, the population evolution is governed by two vectors: the true
class frequencies, and the current empirical frequencies. In each timesteps, we
apply a function $\mathit{step} : \Delta_m \to \Delta_m$ to the empirical
frequencies to get the updated true frequencies; we will assume that the step
function is contractive, i.e., it is $L$-Lipschitz
\[
  \| \mathit{step}(\vec{x}) - \mathit{step}(\vec{y}) \|_1
  \leq L \cdot \| \vec{x} - \vec{y} \|_1
\]
for $L < 1$. Then, we draw $N$ samples from the distribution given by the true
frequencies and update the empirical frequencies. We can model the evolutionary
process as a simple probabilistic program $\PopDyn(T)$ which repeats $T$
iterations of the evolutionary step:
\[
  \begin{array}{l}
    \vec{x} \iass x_0; t \iass 0; \\
    \iwhile{t < T}{} \\
    \quad \vec{p} \iass \mathit{step}(\vec{x});\\
    \quad \vec{x} \iass \vec{0}; j \iass 0; \\
    \quad \iwhile{j < N}{} \\
    \quad\quad \vec{z} \irnd \multD(\vec{p}); \\
    \quad\quad \vec{x} \iass \vec{x} + (1/N) \cdot \vec{z} ; \\
    \quad\quad j \iass j + 1 ; \\
    \quad t \iass t + 1
  \end{array}
\]
The vector $\vec{x}$ stores the current empirical frequencies (the distribution
of each class in our current population), while the vector $\vec{p}$ represents
the true frequencies for the current step.

\begin{figure*}
  \[
    \inferrule[Mult-Max]
    { }
    { \Peprhl
      {\vec{x}\lside \irnd \multD(\vec{p}\lside)}
      {\vec{x}\rside \irnd \multD(\vec{p}\rside)}
      { \top; \| \vec{p}\lside - \vec{p}\rside \|_1 }
      { \vec{x}\lside, \vec{x}\rside \in \{ 0, 1 \}^m ; \| \vec{x}\lside - \vec{x}\rside \|_1 }
      {\id} }
  \]
  \caption{\label{fig:mult-opt} Maximal coupling rule for multinomial}
\end{figure*}

% \gb{explain better what is the multinomial distribution. The
%   interesting property that points at distance $<1/N$ are equal also
%   needs to be restored and explained}
 
In the sampling instruction, $\multD(\vec{p})$ represents the
\emph{multinomial distribution} with parameters $\vec{p}$; this distribution can
be thought of as generalizing a Bernoulli (biased coin toss) distribution to $m$
outcomes, where each outcome has some probability $p_i$ and $\sum_i p_i =
1$. We represent samples from the multinomial distribution as binary vectors in
$\Delta_m$: with probability $p_i$, the sampled vector has the $i$th entry set
to $1$ and all other entries $0$.

To analyze the sampling instruction, we introduce the rule \rname{Mult-Max} in
\cref{fig:mult-opt}. This rule encodes the \emph{maximal coupling}---a standard
coupling construction that minimizes the probability of returning different
samples---of two multinomial distributions; in the \LONGVERSION, we show that this
rule is sound. The post-condition $\vec{x}\lside, \vec{x}\rside \in \{ 0, 1
\}^m$ states that the samples are always binary vectors of length $m$, while the
distances indicate that the expected distance between the sampled vectors $\|
\vec{x}\lside - \vec{x}\rside \|_1$ is at most the distance between the
parameters $\| \vec{p}\lside - \vec{p}\rside \|_1$.

Given two possibly different initial frequencies $(x_0)\lside, (x_0)\rside \in
\Delta_m$, we want to show that the resulting distributions on empirical
frequencies from $\PopDyn(T)$ converge as $T$ increases. We will construct an
expectation coupling where the expected distance between the empirical
distributions $x\lside$ and $x\rside$ decays exponentially in the number of
steps $T$; by \cref{prop:conseq:tv}, this implies that the total-variation
distance between
the distributions of $x\lside$ and $x\rside$ decreases exponentially quickly.
Formally, we prove the following judgement:
\begin{equation} \label{eq:popdyn-whole}
  \Peprhl{\PopDyn(T)}{\PopDyn(T)}
  {\pre; \| (\vec{x_0})\lside - (\vec{x_0})\rside \|_1}
  {\pre; \| \vec{x}\lside - \vec{x}\rside \|_1}
  {\scale{L^T}}
\end{equation}
where
\[
  \pre \triangleq
  \| \vec{x} \lside - \vec{x} \rside \|_1 < 1/N
  \implies \vec{x} \lside = \vec{x} \rside .
\]
$\pre$ is an invariant throughout because every entry of
$\vec{x} \lside$ and $\vec{x} \rside$ is an integer multiple of $1/N$.

To prove the inner judgment, let $s_{\mathit{out}}$ and $s_{\mathit{in}}$ be the
outer and inner loops, and let $w_{\mathit{out}}$ and $w_{\mathit{in}}$ be their
loop bodies. We proceed in two steps. In the inner loop, we want
\begin{equation} \label{eq:popdyn:s-in}
  \Peprhl
  { s_{\mathit{in}} }
  { s_{\mathit{in}} }
  { \pre ; \| \vec{p}\lside - \vec{p}\rside \|_1 }
  { \pre ; \| \vec{x}\lside - \vec{x}\rside \|_1 }
  {\id}
\end{equation}
hiding invariants asserting $j$ and $t$ are equal in both runs. By the loop rule
\rname{While}, it suffices to prove
\begin{equation} \label{eq:popdyn:w-in}
  \Peprhl
  { w_{\mathit{in}} }
  { w_{\mathit{in}} }
  { e\lside = k \land \pre ; \preexp_k }
  { e\lside = k - 1 \land \pre ; \preexp_{k - 1} }
  {\id}
\end{equation}
for each $0 < k \leq N$, where $\preexp_k \triangleq \| x\lside - x\rside \|_1 +
(k/N) \cdot \| p\lside - p\rside \|_1$ and the decreasing variant is $e
\triangleq N - j$.  Let the sampling command be $w'_{\mathit{in}}$, and the
remainder of the loop body be $w''_{\mathit{in}}$. By applying the multinomial
rule \rname{Mult-Max} and using the rule of consequence to scale the distances
by $1/N$, we have
\[
  \Peprhl
  {w'_{\mathit{in}}}
  {w'_{\mathit{in}}}
  { \pre; (1/N) \cdot \| \vec{p}\lside - \vec{p}\rside \|_1 }
  { \pre ; (1/N) \cdot \| \vec{z}\lside - \vec{z}\rside \|_1 }
  {\id} .
\]
Since the sampling command does not modify the vectors $\vec{x}, \vec{p}$, we
can add the distance $\preexp_{k-1}$ to the pre-and the post-conditions
by the frame rule \rname{Frame-D} (noting that the distance transformer $\id$ is
non-contractive). Since $\preexp_k = \preexp_{k - 1} + (1/N) \cdot \|
\vec{p}\lside - \vec{p}\rside \|_1$ by definition, we have
\begin{equation} \label{eq:popdyn:wprime-in}
  \Peprhl
  {w'_{\mathit{in}}}
  {w'_{\mathit{in}}}
  { \pre; \preexp_k }
  { \pre ; \preexp_{k-1} + (1/N) \cdot \| \vec{z}\lside - \vec{z}\rside \|_1 }
  {\id} .
\end{equation}
For the deterministic commands $w''_{\mathit{in}}$, the assignment rule \rname{Assg}
gives
\begin{equation} \label{eq:popdyn:w-det-in}
  \Peprhl
  {w''_{\mathit{in}}}
  {w''_{\mathit{in}}}
  { \pre; \preexp_{k-1}[\subst{\vec{x}}{(\vec{x} + (1/N) \cdot \vec{z})}]}
  { \pre ; \preexp_{k-1} }
  {\id} ,
\end{equation}
where the substitution is made on the respective sides.  Applying the rule of
consequence with the triangle inequality in the pre-condition, we can combine
this judgment (\cref{eq:popdyn:w-det-in}) with the judgment for
$w'_{\mathit{in}}$ (\cref{eq:popdyn:wprime-in}) to verify the inner loop body
(\cref{eq:popdyn:w-in}). The rule \rname{While} gives the desired judgment for
the inner loop $s_{\mathit{in}}$ (\cref{eq:popdyn:s-in}).

Turning to the outer loop, we first prove a judgment for the loop bodies:
\[
  \Peprhl
  {w_{\mathit{out}}}
  {w_{\mathit{out}}}
  { \pre ; \| \vec{x}\lside - \vec{x}\rside \|_1 }
  { \pre ; \| \vec{x}\lside - \vec{x}\rside \|_1  }
  {\scale{L}} .
\]
By the sequence and assignment rules and the judgment for the inner loop
(\cref{eq:popdyn:s-in}), we have
\[
  \Peprhl
  {w_{\mathit{out}}}
  {w_{\mathit{out}}}
  { \pre ; \| \mathit{step}(\vec{x}\lside) - \mathit{step}(\vec{x}\rside) \|_1 }
  { \pre ; \| \vec{x}\lside - \vec{x}\rside \|_1  }
  {\id} .
\]
Applying the fact that $\mathit{step}$ is $L$-Lipschitz, the rule of consequence
gives
\[
  \Peprhl
  {w_{\mathit{out}}}
  {w_{\mathit{out}}}
  { \pre ; \| \vec{x}\lside - \vec{x}\rside \|_1 }
  { \pre ; \| \vec{x}\lside - \vec{x}\rside \|_1 }
  {\scale{L}}
\]
for the outer loop body. We can then apply the rule \rname{While} to conclude
the desired judgment for the whole program (\cref{eq:popdyn-whole}).

This judgment shows that the distributions of $\vec{x}$ in the two runs converge
exponentially quickly. More precisely, let $\nu\lside, \nu\rside$ be the
distributions of $\vec{x}$ after two executions of $\PopDyn(T)$ from
initial frequencies $(x_0)\lside, (x_0)\rside \in \Delta_m$.
\cref{eq:popdyn-whole} implies that there is an expectation coupling
\[
  \bcouplingsupp {\nu} {\| \cdot \|_1} {\delta} {\nu\lside} {\nu\rside} {\pre} ,
\]
where $\delta = L^T \cdot \| (x_0)\lside - (x_0)\rside \|_1$. All pairs of
vectors $(v_1, v_2)$ in the support of $\nu$ with $v_1 \neq v_2$ are at distance
at least $1/N$ by the support condition $\Phi$, so \cref{prop:conseq:tv} implies
\[
  \text{TV}(\nu\lside, \nu\rside) \leq N \cdot L^T .
\]
Since $L < 1$, the distributions converge exponentially fast as $T$ increases.

\section{Path Coupling and Graph Coloring} \label{sec:pathcoupling}
Path coupling is a powerful method for proving rapid mixing of Markov
chains~\citep{bubley1997path}. We review the central claim of path coupling from
the perspective of expected sensitivity. Then, we define an extension of our
program logic that incorporates the central idea of path coupling. Finally, we
apply of our logic to verify a classical example using the path coupling method.

\subsection{Path Coupling and Local Expected Sensitivity}
So far, we have assumed very little structure on our distances; essentially they
may be arbitrary non-negative functions from $A \times A$ to the real numbers.
Commonly used distances tend to have more structure. For integer-valued
distances, we can define a weakening of sensitivity that only
considers pairs of inputs at distance $1$, rather than arbitrary
pairs of inputs. We call the resulting property \emph{local} expected sensitivity.
\begin{definition}
Let $\preexp_A$ be an integer-valued distance over $A$ and $\preexp_B$
be a distance over $B$. Moreover, let $f\in\mathcal{L}$. We say that a
probabilistic function $g:A\to\distr(B)$ is \emph{locally expected}
$f$-\emph{sensitive} (with respect to $\preexp_A$ and $\preexp_B$) if
for every $x_1,x_2\in A$ such that $\preexp_A(x_1,x_2)=1$, we have
\[
  \E{(y_1, y_2) \sim \mu}{\preexp_{B}(y_1, y_2)} \leq f(\preexp_A(x_1, x_2)) = f(1)
\]
for some coupling $\coupling {\mu} {g(x_1)} {g(x_2)}$. 
\end{definition}
In general, local expected $f$-sensitivity is weaker than expected
$f$-sensitivity. However, both notions coincide under some mild
conditions on the distances. We introduce a pair of conditions:
\begin{align}
  &\begin{cases}
    \forall x,y.\, \preexp(x, y)= 0 \implies x = y \\
    \forall x,y.\, \preexp(x, y)=n+1 \implies \exists z.\, \preexp(x,z)=1\land \preexp(z,y)=n
  \end{cases}
  \tag{P} \label{cond:path} \\
  &\begin{cases}
    \forall x.\, \preexp(x, x)=0  \\
    \forall x, y, z .\, \preexp(x, z) \leq \preexp(x, y) + \preexp(y, z)
  \end{cases}
  \tag{H} \label{cond:hemi}
\end{align}
In condition~\eqref{cond:path}, $\preexp$ is an
integer-valued distance. The first condition is standard for metrics. The second
condition is more interesting: if two points are at distance $2$ or greater, we
can find a strictly intermediate point. We will soon see an important class of
distances---\emph{path metrics}---that satisfy these conditions
(\cref{def:pathmetric}).  Condition~\eqref{cond:hemi} is more standard: the
distance $\preexp$ should assign distance $0$ to two equal points, and satisfy
the triangle inequality.  Every metric satisfies these properties; in general,
such a distance is called a \emph{hemimetric}.

When the pre-distance satisfies~\eqref{cond:path} and the post-distance
satisfies~\eqref{cond:hemi}, local expected sensitivity is equivalent to
expected sensitivity for linear distance transformers.

\begin{proposition}\label{prop:local:sens}
  Let $\preexp_A$ be an integer-valued distance over $A$ satisfying
  \eqref{cond:path}, and let $\preexp_B$ be a distance over $B$ satisfying
  \eqref{cond:hemi}.  Let $f\in\mathcal{L}$ and $g:A\to\distr(B)$.  Then $g$ is
  locally expected $f$-sensitive iff it is expected $f$-sensitive (both with
  respect to $\preexp_A$ and $\preexp_B$).
\end{proposition}
\begin{proof}
  The reverse direction is immediate.
  The forward implication is proved by induction on
  $\preexp_A(x_1,x_2)$. For the base case, where $\preexp_A(x_1,x_2)=0$, we have
  $x_1=x_2$ and hence $g(x_1)=g(x_2)$. Letting $\mu$ be the identity
  coupling for $g(x_1)$ and $g(x_2)$, we have $\E \mu {\preexp_B} =
  \sum_{y\in B} \preexp_B(y,y)=0$ since $\preexp_B(y,y)=0$ for every $y$,
  establishing the base case. For the
  inductive step, assume that $\preexp_A(x_1,x_2)=n+1$. Then there exists
  $x'$ such that $\preexp_A(x_1,x')=1$ and $\preexp_A(x',x_2)=n$. By induction,
  there exist two expectation couplings $\mu_1$ and $\mu_n$ satisfying the
  distance conditions
  \[
    \E {\mu_1} {\preexp_B}\leq f(\preexp_B(x_1,x'))
    \quad\text{and}\quad
    \E {\mu_n} {\preexp_B}\leq f(\preexp_B(x',x_2)) .
  \]
  Define $\mu$ as
  \[
    \mu(x,y) \triangleq \sum_{z \in A} \frac{\mu_1(x,z) \cdot \mu_n(z,y)}{g(x')(z)} ,
  \]
  where we treat terms with zero in the denominator as $0$; note that since
  $\mu_1$
  and $\mu_n$ satisfy the marginal conditions, we have $\pi_2(\mu_1) =
  \pi_1(\mu_n) = g(x')$, so $g(x')(z) = 0$ implies that $\mu_1(x,z) = \mu_n(z,
  y) = 0$, so the numerator is also zero in these cases.
  
  Now, the marginal conditions $\pi_1(\mu) = g(x_1)$ and $\pi_2(\mu) = g(x_2)$
  follow from the marginal conditions for $\mu_1$ and $\mu_n$. The distance
  condition $\E {\mu} {\preexp_B} \leq f(\preexp_A(x_1,x_2))$ is a bit more involved:
  \begin{align}
    \E {\mu} {\preexp_B}
    &= \sum_{x,y} \mu(x,y) ~\preexp_B(x,y) \notag \\
    &= \sum_{x,y} \sum_{z} \left(\frac{\mu_1(x,z) ~\mu_n(z,y)}{ g(x')(z) }\right)~\preexp_B(x,y)
    \notag \\
    &\leq \sum_{x,y,z} \left(\frac{\mu_1(x,z) ~\mu_n(z,y)}{ g(x')(z) }\right)~\preexp_B(x,z)
        + \sum_{x,y,z} \left(\frac{\mu_1(x,z) ~\mu_n(z,y)}{ g(x')(z) }\right)~\preexp_B(z,y)
    \tag{triangle ineq.} \\
    &= \sum_{y,z} \left( \sum_x \frac{\mu_1(x,z)}{g(x')(z)} \right) ~\mu_n(z,y)~\preexp_B(z,y)
    + \sum_{x,z} \left( \sum_y \frac{\mu_n(z,y)}{g(x')(z)} \right)~ \mu_1(x,z) ~ \preexp_B(x,z)
    \notag \\
    &= \sum_{x,z} \mu_1(x,z) ~ \preexp_B(x,z) + \sum_{y,z} \mu_n(z,y)~\preexp_B(z,y)
    \tag{marginals} \\
    &= \E {\mu_1} {\preexp_B} + \E {\mu_n} {\preexp_B}  \notag \\
    & \leq f(\preexp_A(x_1,x')) + f(\preexp_A(x',x_2))
    \tag{distances} \\
    & = f(\preexp_A(x_1,x') + \preexp_A(x',x_2))
    \tag{$f$ linear} \\
    & = f(\preexp_A(x_1,x_2)) \notag .
  \end{align}
  Thus, we have an expectation coupling $\bcouplingsupp {\mu} {\preexp_B}
  {\delta} {g(x_1)} {g(x_2)} {}$ for $\delta = f(\preexp_A(x_1, x_2))$. This
  completes the inductive step, so $g$ is expected $f$-sensitive.
\end{proof}

One important application of our result is for path metrics.

\begin{definition}[Path metric] \label{def:pathmetric}
  Let $\pre$ be a binary relation over $A$, and let $\pre^*$ denote its
  transitive closure and $\pre^n$ denote the union of its $n$-fold compositions
  for $n \geq 1$. Assume that for every $a,a'\in A$, we have $(a,a') \in
  \pre^*$. The \emph{path metric} of $\pre$ is the distance
\[
  \pathdistance{\pre}(a,a')= \min_n \{ (a, a') \in \pre^n \} 
\]
Note that the set is non-empty by assumption, and hence the minimum
is finite.
\end{definition}
Path metrics evidently satisfy condition \eqref{cond:path}. Since they are also
metrics, they also satisfy condition \eqref{cond:hemi}.  The fundamental theorem
of path coupling is then stated---in our terminology---as follows.
\begin{corollary}[\citet{bubley1997path}]
Let $\preexp=\pathdistance{\pre}$ for a binary relation $\pre$
over $A$. Let $g:A\to \distr(A)$ be a locally expected $f$-sensitive
function, where $f\in\mathcal{L}$. Then for every $T\in\mathbb{N}$,
the $T$-fold (monadic) composition $g^T$ of $g$ is expected
$f^T$-sensitive, i.e.\, for every $x_1,x_2\in A$, there exists a
coupling $\coupling {\mu} {g^{T}(x_1)} {g^{T}(x_2)}$ such that

\[
  \E \mu {\preexp} \leq f^T (\preexp(x_1,x_2)) .
\]
\end{corollary}
\begin{proof}
The proof follows from the equivalence between local expected
sensitivity and sensitivity, and the composition theorem of expected
sensitive functions.
\end{proof}

\subsection{Program Logic}
We formulate a proof rule inspired from
local expected sensitivity, in \cref{fig:pathcoupling}. Let us first consider
the premises of the rule. The first three conditions are inherited from
\cref{prop:local:sens}: the distance transformer $f$ is linear, the pre-distance
$\preexp$ is $\NN$-valued, and the post-distance satisfies condition
\eqref{cond:hemi}. The new two conditions deal with
the pre- and post-conditions, respectively. First, the pre-condition $\pre$ and
the pre-distance $\preexp$ satisfy the following condition:
\begin{align*}
  \pathcompat(\pre, \preexp) \triangleq{}&
  \forall m_1, m_2, n \in \NN.\, \pre(m_1, m_2) \land \preexp(m_1, m_2) = n + 1
  \\
  &\implies \exists m'.\, \preexp(m_1, m') = 1 \land \preexp(m', m_2) = n \land
  \pre(m_1, m') \land \pre(m', m_2) .
\end{align*}
This condition implies that $\preexp$ satisfies condition \eqref{cond:path}
(needed for \cref{prop:local:sens}), but it is stronger: when the distance
$\preexp$ is at least $1$, we can find some memory $m'$ such that the
pre-condition can also be split into $\pre(m_1, m')$ and $\pre(m', m_2)$.  We
call this condition \emph{path compatibility}; intuitively, it states that the
pre-condition is compatible with the path structure on the pre-distance.
Likewise, the post-condition $\post$ must be transitively closed; the
transitivity rule represents a finite sequence of judgments with
post-condition $\post$.

The main premises cover two cases: either the initial memories
are at distance $0$, or they are at distance $1$. Given these two judgments, the
conclusion gives a judgment for two input memories at \emph{any} distance. In
this way, the rule \rname{Trans} models a transitivity principle for expectation
couplings.

\begin{theorem}[Soundness] \label{thm:soundness:trans}
The rule \rname{Trans} is sound: for every instance of the rule concluding
$\eprhl{s_1}{s_2}{\pre;\preexp}{\post;\postexp}{f}$ and initial memories
satisfying $(m_1, m_2) \models \pre$, there exists $\mu$ such that
$$\bcouplingsupp{\mu}{\postexp}{f(\preexp(m_1, m_2))}{\dsem{m_1}{s_1}}{
  \dsem{m_2}{s_2}}{\post} .$$
\end{theorem}
\begin{proof}
  By a similar argument as \cref{prop:local:sens}, with careful handling for the
  pre- and post-conditions. We defer details to the \LONGVERSION.
\end{proof}

\begin{figure}
\[
  \inferrule[Trans] {
    f \in \mathcal{L} \\
    \preexp : \NN \\
    \postexp \text{ satisfies } \eqref{cond:hemi}
    \\\\
    \models \pre \implies \pathcompat(\preexp, \pre) \\
    \models \post^* \implies \post
    \\\\
    \Peprhl{s}{s}{\pre\land\preexp=0;-}{\post;\postexp}{0} \\
    \Peprhl{s}{s}{\pre\land\preexp=1; -}{\post; \postexp}{\const{f(1)}}} {
    \Peprhl{s}{s}{\pre;\preexp}{\post; \postexp}{f} }
  \]
  \caption{Transitivity rule}
\label{fig:pathcoupling}
\end{figure}

\subsection{Example: Glauber Dynamics}
The \emph{Glauber dynamics} is a randomized algorithm for
approximating uniform samples from the valid colorings of a finite
graph. It is a prime example of an algorithm where rapid mixing can be
established using the path coupling method~\citep{bubley1997path}.

Before detailing this example, we recall some basic definitions and notations.
Consider a graph $G$ with a finite set of \emph{vertices} $V$ and a symmetric relation
$E \subseteq V \times V$ representing the \emph{edges}, and let $C$ be a finite set of
\emph{colors}. A \emph{coloring} of $G$ is a map $w : V \to C$; a coloring is
\emph{valid} if neighboring vertices receive different colors: if $(a, b) \in
E$, then $w(a) \neq w(b)$. We write $w(V')$ for the set of colors at a set of
vertices $V' \subseteq V$.

For a graph $G$ and a fixed set of colors $C$, there may be multiple (or perhaps
no) valid colorings. \citet{DBLP:journals/rsa/Jerrum95} proposed a simple Markov
chain for sampling a uniformly random coloring.  Beginning at any coloring $w$,
it draws a uniform vertex $v$ and a uniform color $c$, and then changes the
color of $v$ to $c$ in $w$ if no neighbor of $v$ is colored $c$. The Glauber
dynamics repeats this process for $T$ steps $T$ and returns the final coloring.
We can model this process with the following program $\glauber(T)$:
\[
  \begin{array}{l}
    i \iass 0; \\
    \iwhile{i < T}{} \\
    \quad {v}\irnd {V}; \\
    \quad {c}\irnd {C}; \\
    \quad \ift{\valid{G}(w, v, c)}{w \iass w[v \mapsto c]}; \\
    \quad i \iass i+1; \\
    \iret{w}
  \end{array}
\]
The guard $\valid{G}(w, v, c)$ is true when the vertex $v$ in coloring $w$ can
be colored $c$. \citet{DBLP:journals/rsa/Jerrum95} proved that the distribution
over outputs for this process converges rapidly to the uniform distribution on
valid colorings of $G$ as we take more and more steps, provided we start with a
valid coloring. While the original proof was quite technical,
\citet{bubley1997path} gave a much simpler proof of the convergence by
\emph{path coupling}.

Roughly, suppose that for every two colorings that differ in exactly \emph{one}
vertex coloring, we can couple the distributions obtained by executing one step
of the transition function of the Markov process (i.e., the loop body above)
such that the expected distance (how many vertices are colored differently) is
at most $\beta$. Then, the path coupling machinery gives a coupling of the
processes started from any two colorings, and concludes that after $T$ steps the
expected distance between two executions started with colorings at distance $k$
is upper bounded by $\beta^T \cdot k$.

In \SYSTEM, this final property corresponds to the following judgment:
\[
  \Peprhl {\glauber(T)}{\glauber(T)} {\pre_G;\pathdistance{\Adj}}
          {\top; \pathdistance{\Adj}} {\scale{\beta^T}}
\]
Above, $\Adj$ holds on two states iff the colorings (stored in
the variable $w$) differ in the color of a single vertex, and $\postexp \eqdef
\pathdistance{\Adj}$ counts the number of vertices with $w\lside(v) \neq
w\rside(v)$. The pre-condition $\pre_G$ captures properties of the graph; in
particular, $\pre_G$ states that $\Delta$ is the maximal degree in
$G$, i.e., each vertex in $G$ has at most $\Delta$ neighbors. Finally, $\beta$
is a constant determined by the graph and the number of colors; in certain
parameter ranges, $\beta$ is
strictly less than $1$ and the Markov chain converges quickly from any initial
state.

Now, we present the proof. Since the graph $G$ is not modified in the program,
we keep $\pre_G$ as an implicit invariant throughout. We begin with the loop
body $s$. We apply the rule \rname{Trans} with pre- and post-condition $\Phi,
\Psi \triangleq i\lside = i\rside$, distances $\preexp, \postexp \triangleq
\pathdistance{\Adj}$, and $f \triangleq \scale{\beta}$. The side-conditions are
clear: $f$ is linear, the $\postexp$ satisfies condition \eqref{cond:hemi}, the
pre-condition $\pre$ is compatible with the path distance $\preexp$, and the
post-condition $\post$ is transitively closed.  The first main premise
\[
  \Peprhl {s} {s}
  {\Phi \land \pathdistance{\Adj} = 0; -}
  {\Psi; \pathdistance{\Adj}}
  {0}
\]
is easy to show: the initial states have $w\lside = w\rside$, so simply coupling
using the identity bijection in \rname{Rand} preserves $w\lside = w\rside$ and
keeps the states at distance $0$. The second main premise
\[
  \Peprhl {s} {s}
  {\Phi \land \pathdistance{\Adj} = 1; -}
  {\Psi; \pathdistance{\Adj}}
  {\const{\beta}} ,
\]
is more complicated. Note that $\pathdistance{\Adj} = 1$ is equivalent to
$\Adj$: the two initial colorings $w\lside$ and $w\rside$ must differ in the
color at a single vertex. So, $\Adj$ implies the invariant
\[
  \Xi(a, b, v_\delta) \triangleq \forall z \in V.\,
  \begin{cases}
    z = v_\delta \implies a = w\lside(z) \land b = w\rside(z) \\
    z \neq v_\delta \implies w\lside(z) = w\rside(z)
  \end{cases}
\]
for some differing vertex $v_\delta$, which is colored as $a \triangleq
w\lside(v_\delta)$ and $b \triangleq w\rside(v_\delta)$ in the two respective
colorings. By the \rname{Case} rule, it suffices to show
\[
  \Peprhl {s} {s}
  {\Phi \land \Xi(a, b, v_\delta); -}
  {\Psi; \pathdistance{\Adj}}
  {\const{\beta}}
\]
for every $a, b \in C$ and $v_\delta \in V$.

We apply the \rname{SeqCase} rule with 
$s_{\mathit{samp}}$, consisting of the two random samplings in the loop body,
and $s_{\mathit{rest}}$, consisting of the
conditional statement and the updates.  For the first judgment, we first couple
the vertex samplings with the identity coupling so that $v\lside = v\rside$,
using the rule \rname{Rand} with $h = \id$. This gives:
\[
  \Peprhl {v \irnd V} {v \irnd V}
  {\Phi \land \Xi(a, b, v_\delta); - }
  {\Phi \land \Xi(a, b, v_\delta) \land v\lside = v\rside; \pathdistance{\Adj}}
  {\const{\beta}} .
\]
Next, we can perform a case analysis on $v\lside$ using the rule \rname{Case}.
If $v\lside$ is not a neighbor of $v_\delta$, then we couple samplings so that
$c\lside = c\rside$ with \rname{Rand} with $h = \id$. Otherwise, we couple
$c\lside = \pi^{ab}(c\rside)$, where $\pi^{ab}$ swaps $a$ and $b$ and leaves all
other colors unchanged. This gives
\[
  \Peprhl {s_{\mathit{samp}}} {s_{\mathit{samp}}}
  {\Phi \land \Xi(a, b, v_\delta); - } {\Theta ; \pathdistance{\Adj}}
  {\const{\beta}} ,
\]
where
\[
  \Theta \triangleq \Phi \land \Xi(a, b, v_\delta) \land v\lside = v\rside
  \land \begin{cases}
    v\lside \in \neighbors{G}(v_\delta) \implies c\lside = \pi^{ab}(c\rside) \\
    v\lside \notin \neighbors{G}(v_\delta) \implies c\lside = c\rside .
  \end{cases}
\]
Continuing with \rname{SeqCase}, we distinguish the following three mutually
exclusive cases with probabilities $q_b$, $q_g$, and $q_n$, depending on how the
distance changes under the coupling:
\begin{itemize}
  \item In the \emph{bad case}, the distance may grow to $2$. Taking the guard
    \[
      e_b \triangleq v \in \neighbors{G}(v_\delta) \land c = b ,
    \]
    the assignment and consequence rules give
    \[
      \Peprhl {s_{\mathit{rest}}} {s_{\mathit{rest}}}
      {\Theta \land {e_b}\lside ; \pathdistance{\Adj}} { \Psi ; \pathdistance{\Adj} } {\scale{2}} .
    \]
    (In fact, this judgment can be proved without the guard ${e_b}\lside$ in the
    pre-condition since the path distance increases from $1$ to at most $2$, but
    we will need to bound the probability of the guard being true in order to
    apply \rname{SeqCase}.) The probability of this case is at most $\Delta/|V|
    |C|$ since we must select a neighbor of $v_\delta$ and the color $b$ in the
    first side, so $q_b \leq \Delta / |V| |C|$.
  \item In the \emph{good case}, the distance shrinks to zero. We take the guard
    \[
      e_{g} \triangleq v = v_\delta \land c \notin w(\neighbors{G}(v)) .
    \]
    By applying the assignment and consequence rules, we can prove:
    \[
      \Peprhl {s_{\mathit{rest}}} {s_{\mathit{rest}}}
      {\Theta \land {e_g}\lside ; \pathdistance{\Adj}} { \Psi ; \pathdistance{\Adj}} {\scale{0}} .
    \]
    We will later need a \emph{lower} bound on the probability of this case:
    since we must choose the differing vertex $v_\delta$ and a color different
    from its neighbors, and there are at most $\Delta$ neighbors, $q_g \geq (|C|
    - \Delta) / |C| |V|$.
  \item In the \emph{neutral case}, we take the guard  $e_n \triangleq \neg e_b
    \land \neg e_g$. The assignment rule gives
    \[
      \Peprhl {s_{\mathit{rest}}} {s_{\mathit{rest}}}
      {\Theta \land {e_n}\lside ; \pathdistance{\Adj}} { \Psi ; \pathdistance{\Adj}} {\id} ,
    \]
    showing that the distance remains unchanged.
\end{itemize}
To put everything together, we need to bound the average change in distance.
Since the cases are mutually exclusive and at least one case holds, we know $q_n
= 1 - q_b - q_g$.  Combining the three cases, we need to bound the function $x
\mapsto (q_n + 2 \cdot q_b) \cdot x = (1 - q_g + q_b) \cdot x$. By the upper
bound on $q_b$ and the lower bound on $q_g$, \rname{SeqCase} gives
\[
  \Peprhl
  {s_{\mathit{samp}}; s_{\mathit{rest}}}
  {s_{\mathit{samp}}; s_{\mathit{rest}}}
  {\Phi \land \Xi(a, b, v_\delta); -}
  {\Psi; \pathdistance{\Adj}} {\beta} ,
\]
for every $a, b \in C$ and $v_\delta \in V$,where
\[
  \beta \eqdef 1 - \frac{1}{|V|} + \frac{2 \Delta}{|C| |V|} .
\]
So, we also have
\[
  \Peprhl
  {s_{\mathit{samp}}; s_{\mathit{rest}}}
  {s_{\mathit{samp}}; s_{\mathit{rest}}}
  {\Phi \land \pathdistance{\Adj} = 1; -}
  {\Psi; \pathdistance{\Adj}} {\beta}
\]
and the rule \rname{Trans} gives
\[
  \Peprhl
  {s_{\mathit{samp}}; s_{\mathit{rest}}}
  {s_{\mathit{samp}}; s_{\mathit{rest}}}
  {\Phi; \pathdistance{\Adj}}
  {\Psi; \pathdistance{\Adj}} {\beta} .
\]
Finally, we apply the rule \rname{While} with invariant $\Phi = \Psi$ and the
assignment rule \rname{Assg} to conclude the desired judgment
\[
  \Peprhl {\glauber(T)}{\glauber(T)} {\pre_G; \pathdistance{\Adj}}
          {\top; \pathdistance{\Adj}} {\scale{\beta^T}} .
\]
When the number of colors $|C|$ is strictly larger than $2\Delta$, the constant
$\beta$ is strictly less than $1$ and the Glauber dynamics is rapidly mixing.

\section{Prototype Implementation} \label{sec:implem}
We have developed a prototype implementation of our program logic on
top of \textsf{EasyCrypt}, a general-purpose proof assistant for
reasoning about probabilistic programs, and formalized stability of
the convex version of Stochastic Gradient Method and convergence of
population dynamics and Glauber dynamics.

\begin{itemize}

\item For some rules, we implement stronger versions that are required
  for formalization of the examples. For instance, our implementation
  of the \rname{Conseq} rule supports scaling of distances.

%% For \rname{Conseq}, the rule we implement is (a generalization of):
%%   %
%%   $$
%% \inferrule[Conseq]
%%  { \Peprhl{s_1}{s_2}{\pre; \preexp}{\post; \postexp}{f} \\
%%    \models \pre' \implies \pre \\
%%    \models \post \implies \post' \\
%%     \alpha > 0\\\\
%%    \forall m_1, m_2 \models \pre' .\, \alpha\cdot f(\preexp(m_1, m_2)) \leq
%%    f'(\preexp''(m_1, m_2)) \\
%%    \forall m_1, m_2 \models \post .\, \postexp''(m_1, m_2) \leq
%%    \alpha\cdot\postexp(m_1, m_2) }
%%  { \Peprhl{s_1}{s_2}{\pre';\alpha\cdot\preexp''}{\post'; \alpha\cdot\postexp''}{f'} }
%%  $$
%%  %
%% The main difference with the rule from \cref{fig:rules} is the
%%  multiplicative factor $\alpha$, allowing to encode scalings of 
%%  distances.

\item The ambient higher-order logic of \textsf{EasyCrypt} is used
  both for specifying distributions and for reasoning about
  their properties. Likewise, the logic is used for defining distances,
  Lipschitz continuity, and affine functions, and for proving their
  basic properties.

\item We axiomatize the gradient operator and postulate its main
  properties. Defining gradients from first principles and proving
  their properties is technically possible, but beyond the scope of
  the paper. Similarly, we axiomatize norms and state relevant
  properties as axioms. A small collection of standard facts are
  assumed.

\end{itemize}
The formalization of the examples is reasonably straightforward. The
formalization of stability for the Stochastic Gradient Method is about
400 lines; about one third is devoted to proving mathematical
facts. The formalization of convergence for the population dynamics about is
150 lines, while formalization of convergence for the Glauber dynamics is about
550 lines.

We have not yet interfaced current the prototype with the rich set of program
transformations supported by \textsf{EasyCrypt}, e.g.\ code motion,
loop unrolling, loop range splitting, which are required for the
non-convex version of Stochastic Gradient Method. Implementing these
features should not pose any difficulty, and is left for future work.

\section{Related Work} \label{sec:rw}
% There is a long tradition of using non-expansive ($\alpha$-sensitive
% functions, with $\alpha <1$) maps over metric spaces for defining the
% denotational semantics of deterministic and probabilistic programs;
% e.g.,~\citep{ArnoldN80,DeBakkerZ82,vanbreugel,AmorimGHKC17,Kozen79}.
% It is also common to interpret programs as
% functions between ultrametric spaces, a special class of metric spaces
% where $\max$ is used instead of addition in the triangular inequality;
% e.g.~\citep{MacQueenPS84,AmericaR87,AbadiP90,KrishnaswamiB11}.

Lipschitz continuity has also been considered extensively in the
setting of program verification: \citet{ChaudhuriGL10} develop a
SMT-based analysis for proving programs robust, in the setting of a
core imperative language; \citet{ReedPierce10} develop a linear type system
for proving sensitivity and differential privacy in a higher-order
language~\citep{GHHNP13,AGGH14,AmorimGHKC17,adafuzz}.

There is also a long tradition of verifying expectation properties of
probabilistic programs; seminal works include \Sppdl~\citep{Kozen85}
and \Spgcl~\citep{Morgan96}. Recently, \citet{KaminskiKMO16} have
developed a method to reason about the expected
running time of probabilistic programs. This line of work is focused on
non-relational properties, such as proving upper bounds on errors,
whereas expected sensitivity is intrinsically relational.

There has also been a significant amount of work on the relational
verification of probabilistic programs. Barthe and collaborators
develop relational program logics for reasoning about the provable
security of cryptographic constructions~\citep{BartheGZ09} and
differential privacy of algorithms~\citep{BartheKOZ12}. \SYSTEM
subsumes the relational program logic considered
by~\citet{BartheGZ09}; indeed, one can prove that the two-sided rules
of \Sprhl are essentially equivalent to the fragment of \SYSTEM where
the pre-distance and post-distance are the zero function. In
contrast, the relational program logic \Saprhl considered
by~\citet{BartheKOZ12} and developed in subsequent
work~\citep{BartheO13,DBLP:conf/lics/BartheGGHS16,DBLP:conf/ccs/BartheFGGHS16,sato2016approximate,BEHSS17,JHThesis}
is not comparable with \SYSTEM. \Saprhl uses a
notion of approximate coupling targeting
differential privacy, while expectation couplings are
designed for average versions of quantitative
relational properties. In particular, \Saprhl considers
\emph{pointwise} notions of distance between distributions without
assuming a distance on the sample space, while \SYSTEM works
with distances on the underlying space, proving fundamentally
different properties.

There have been a few works on more specific
relational expectation properties. For instance, the
standard target property in \emph{masking} implementations in
cryptography is a variant of probabilistic non-interference, known as
\emph{probing security}.  Recent work introduces quantitative masking
strength~\citep{DBLP:journals/tcad/EldibWTS15}, a quantitative
generalization that measures average leakage of the programs.
Similarly, the bounded moment
model~\citep{DBLP:journals/iacr/BartheDFGSS16} is a qualitative,
expectation-based non-interference property for capturing security of
parallel implementations against differential power analyses. Current
verification technology for the bounded moment model is based on a
meta-theorem which reduces security in the bounded moment model to
probing security, and a custom program logic for
proving probing security. It would be interesting to develop a
program logic based on \Seprhl to verify a broader class of parallel
implementations.

For another example, there are formal verification techniques for
verifying \emph{incentive properties} in mechanism design. These
properties are relational, and when the underlying mechanism is
randomized (or when the inputs are randomized), incentive properties
compare the expected payoff of an agent in two
executions. \citet{BartheGGHRS15,BartheGAHRS16} show how to use a
relational type system to verify these properties. While their
approach is also based on couplings, they reason about expectations
only at the top level, as a consequence of a particular coupling. In
particular, it is not possible to compose reasoning about expected
values like in \SYSTEM, and it is also not possible to carry the
analyses required for our examples.

Lastly, \citet{BartheGHS17} use \Sxprhl, a proof-relevant variant of
\Sprhl, to extract a product program for the Glauber dynamics. In a
second step, they analyze the product program to prove rapid mixing;
their analysis is performed directly on the semantics of the product
program. Our system improves upon this
two-step approach in two respects. First, we can internalize the path
coupling principle as a rule in our logic. Second, the probabilistic
reasoning in our system is confined to the side-condition in the
\rname{SeqCase} rule.

\section{Conclusion} \label{sec:conclusion}
We have introduced the notion of expected $f$-sensitivity for
reasoning about algorithmic stability and convergence of
probabilistic processes, and proved some of its basic properties.
Moreover, we have introduced expectation couplings for
reasoning about a broader class of relational expectation properties,
and proposed a relational program logic for proving such properties.
We have illustrated the expressiveness of the logic with recent and
challenging examples from machine learning, evolutionary biology, and
statistical physics.

There are several directions for future work. On the foundational
side, it would be interesting to develop semantic foundations for
advanced fixed point-theorems and convergence criteria that arise in
probabilistic analysis. There are a wealth of results to consider, for
instance, see the survey by \citet{bharucha1976fixed}. On the
practical side, it would be interesting to formalize more advanced
examples featuring relational and probabilistic analysis, like the
recent result by \citet{Shamir16} proving convergence of a
practical variant of the Stochastic Gradient Method, or algorithms for
regret-minimization in learning theory and algorithmic game
theory. Another goal would be to verify more general results about
population dynamics, including the general case
from~\citet{PanageasSV16}.

% Acknowledgments
\begin{acks}                            %% acks environment is optional
                                        %% contents suppressed with 'anonymous'
  %% Commands \grantsponsor{<sponsorID>}{<name>}{<url>} and
  %% \grantnum[<url>]{<sponsorID>}{<number>} should be used to
  %% acknowledge financial support and will be used by metadata
  %% extraction tools.
  % This material is based upon work supported by the
  % \grantsponsor{GS100000001}{National Science
  %   Foundation}{http://dx.doi.org/10.13039/100000001} under Grant
  % No.~\grantnum{GS100000001}{nnnnnnn} and Grant
  % No.~\grantnum{GS100000001}{mmmmmmm}.  Any opinions, findings, and
  % conclusions or recommendations expressed in this material are those
  % of the author and do not necessarily reflect the views of the
  % National Science Foundation.
  We thank the anonymous reviewers for useful comments on this work. This work
  is partially supported by the
  \grantsponsor{GS100010663}{European Research Council}{http://dx.doi.org/10.13039/100010663}
  under Grant No.~\grantnum{GS100010663}{679127},
  the \grantsponsor{GS100000144}{National Science Foundation CNS}{http://dx.doi.org/10.13039/100000144}
  under Grant No.~\grantnum{GS100000144}{1513694}, and
  the \grantsponsor{GS100000893}{Simons Foundation}{http://dx.doi.org/10.13039/100000893}
  under Grant No.~\grantnum{GS100000893}{360368} to Justin Hsu.
\end{acks}

\bibliographystyle{ACM-Reference-Format}
\bibliography{header,main}

\iffull
\clearpage
\onecolumn

\allowdisplaybreaks
\appendix

\section{Soundness}

First, we can show that the equivalence judgment $\eqsem{\pre}{s_1}{s_2}$ shows
that programs $s_1, s_2$ have equal denotation under any memory satisfying the
(non-relation) pre-condition $\pre$.

\begin{lemma} \label[lemma]{l:struct}
  If $\eqsem{\pre}{s_1}{s_2}$, then for any $m \models \pre$,
  $\dsem{m}{s_1} = \dsem{m}{s_2}$.
\end{lemma}

\begin{proof}
  Direct induction on $\eqsem{\pre}{s_1}{s_2}$, using the semantics in
  \cref{fig:semantics} for the base cases.
\end{proof}

Next, we prove the key lemma showing composition of expectation couplings
(\cref{prop:seqcomp}).

\begin{proof}[Proof of composition of expectation couplings (\cref{prop:seqcomp})]
  We check each of the conditions in turn.
  
  For the support condition,  $\supp(\mu) \subseteq \pre$ by the first premise,
  and for every $(a, b) \in \Phi$ we have $\supp(M(a, b)) \subseteq \post$
  by the second premise, so
  \[
    \supp(\mu') = \supp(\E {\mu} {M}) \subseteq \post .
  \]
  For the marginal condition, we have $\pi_1(\mu) = \mu_a$ and $\pi_2(\mu) =
  \mu_b$ by the first premise, and for every $(a, b) \in \Phi$ we have
  $\pi_1(M(a, b)) = M_a(a)$ and $\pi_2(M(a, b)) = M_b(b)$ by the second
  premise. We can directly calculate the marginals of $\mu'$. For instance, for
  every $a' \in A$ the first marginal is
  \begin{align}
    \pi_1(\mu')(a') &= \pi_1(\E {\mu} {M})(a') \notag \\
    &= \sum_{b' \in B} \sum_{(a, b) \in A \times B}
    \mu(a, b) \cdot M(a, b)(a', b') \notag  \\
    &= \sum_{b' \in B} \sum_{(a, b) \in \pre}
    \mu(a, b) \cdot M(a, b)(a', b') \tag{Support of $\mu$} \\
    &= \sum_{(a, b) \in \pre} \mu(a, b) \cdot \pi_1(M(a, b))(a')
    \notag \\
    &= \sum_{(a, b) \in \pre} \mu(a, b) \cdot M_a(a)(a')
    \tag{Marginal of $M(a, b)$} \\
    &= \sum_{a \in A} M_a(a)(a') \sum_{b \in B} \mu(a, b)
    \tag{Support of $\mu$} \\
    &= \sum_{a \in A} M_a(a)(a') \cdot \mu_a(a)
    \tag{Marginal of $\mu$} \\
    &= \mu_a(a') . \notag
  \end{align}
  The second marginal $\pi_2(\mu') = \mu_b$ is similar.
  
  Finally, we check the distance condition. By the premises, we have
  \begin{align*}
    \E {\mu} {\preexp} &\leq \delta \\
    \E {M(a, b)} {\postexp} &\leq f(\preexp(a, b))
    \quad \text{for every } (a, b) \in \pre .
  \end{align*}
  Then, we can bound
  \begin{align}
    \E {\mu'} {\postexp} &= \sum_{(a', b') \in A \times B} \postexp(a', b') \cdot \mu'(a', b')
    \notag \\
    &= \sum_{(a', b') \in A \times B} \postexp(a', b')
    \sum_{(a, b) \in A \times B} \mu(a, b) \cdot M(a, b)(a', b')
    \notag \\
    &= \sum_{(a, b) \in \pre} \mu(a, b) \E {M(a, b)} {\postexp}
    \tag{Support of $\mu$} \\
    &\leq \sum_{(a, b) \in \pre} \mu(a, b) \cdot f(\preexp(a, b))
    \tag{Expectation of $M(a, b)$} \\
    &= \E {\mu} {f(\preexp)}
    \tag{Support of $\mu$} \\
    &\leq f(\E {\mu} {\preexp} )
    \tag{Linearity of expectation} \\
    &\leq f(\delta)
    \tag{Monotonicity of $f$, expectation of $\mu$} .
  \end{align}
\end{proof}

We now move to the soundness of the logic.

\begin{proof}[Proof of the soundness of the logic]
We prove that each rule is sound.

\begin{description}
%\item[\rname{Skip}] Immediate.

%\item[\rname{False}] Immediate

\item[\rname{Conseq}]
  Let $m_1, m_2 \models \pre''$. Hence, $m_1, m_2 \models \pre$, and there
  exists $\eta$ such that
  $\bcouplingsupp {\eta} {\postexp} {\delta}
      {\dsem{m_1}{s_1}} {\dsem{m_2}{s_2}} {\post}$.
  We use $\eta$ for the coupling of the conclusion. We already
  know that
  $\forall i \in \{1, 2\} .\, \proj_i(\mu) = \dsem{m_i}{s_i}$
  and that
  $\supp(\mu) \subseteq \post \subseteq \post''$.
  Finally,
  \begin{align*}
    \E \mu {\postexp''}
      & \leq \E \mu {\postexp} & \text{($\Exp$ monotone)} \\
      & \leq f(\preexp(m_1, m_2)) & \text{($\mu$ coupling)} \\
      & \leq f'(\preexp''(m_1, m_2)) & \text{(premise)}. \\
  \end{align*}

\item[\rname{Struct}]
  Immediate consequence of~\cref{l:struct}.

\item[\rname{Assg} \& \rname{Assg-L}]
  Immediate.

\item[\rname{Rand}]
  Let
  $m_1, m_2 \models \forall v \in \supp(g_1) .\,
  \post [\subst{x_1, x_2}{v, h(v)}]$
  and
  $\mu_i \eqdef \E {v \sim g_i} {\dunit{m_i[\subst{x_i}{v}]}}$
  for $i \in \{1, 2\}$.
  Since $h$ is a one to one mapping from $\supp(g_1)$ to $\supp(g_2)$
  that preserves the mass, we have
  $\wt{\mu_1} = \wt{\mu_2}$
  and
  $\mu_2 = \E {v \sim g_1} {\dunit{m_2[\subst{x_2}{h(v)}]}}$.
  Let
  \[\mu \eqdef \E {v \sim g_1}
      {\dunit{(m_1[\subst{x_1}{v}], m_2[\subst{x_2}{h(v)}])}}.\]
  By construction, for $i \in \{1, 2\}$, we have
  $\proj_i(\mu) = \mu_i$.
  Let $\ov{m} \in \supp(\mu)$. By definition, there exists
  $v \in \supp(g_1)$ s.t.
  $\ov{m} = (m_1[\subst{x_1}{v}], m_2[\subst{x_2}{h(v)}])$.
  Hence, $\ov{m} \models \post$ and
  \begin{align*}
    \E {\ov{m} \sim \mu} {\postexp}
      &= \E {v \sim g_1}
           {\E {\ov{m} \sim \dunit{(m_1[\subst{x_1}{v}], m_2[\subst{x_2}{h(v)}])}}
               {\postexp}} \\
      &= \E {v \sim g_1} {\postexp(m_1[\subst{x_1}{v}], m_2[\subst{x_2}{h(v)}])} \\
      &= \E {v \sim g_1} {\postexp[\subst{(x_1)\lside, (x_2)\rside}{v, h(v)}]}.
  \end{align*}

\item[\rname{Seq}] Let $(m_1, m_2) \models \pre$ and, for
  $i \in \{1, 2\}$, let $\mu_i \eqdef \dsem{m_i}{c_i}$ and
  $\eta_i(m) \eqdef \dsem{m}{c'_i}$.
  From the first premise, we know that there exists an $\eta$ such that
  $\bcoupling \eta {\postexp} \delta {\mu_1} {\mu_2}$
  and $\supp(\eta) \models \Xi$, where
  $\delta \eqdef f(\preexp(m_1, m_2))$.
  Likewise, from the second premise, for
  $m \eqdef (m'_1, m'_2) \models \Xi$,
  there exists an $\eta_{m}$ such that
  $\bcoupling {\eta_{m}} {\postexpz} {\delta'(m)} {\eta_1(m)} {\eta_2(m)}$
  and $\supp(\eta_{m}) \models \post$, where
  $\delta'(m) \eqdef f'(\postexp(m))$.

  \medskip

  Let $\mu \eqdef {\dlet m {\eta} {\eta_m \mid \Xi}}$.
  By \cref{prop:seqcomp}, we already know that
  $\coupling \mu {\dslet {\mu_1} {\eta_1}} {\dslet {\mu_2} {\eta_2}}$
  and that $\supp(\mu) \models \post$. We are left to prove
  that $\E {\mu} {\postexpz} \leq (f' \circ f)(\preexp(m_1, m_2))$:
  \begin{align*}
    \E {\mu} {\postexpz}
      &= \dlet m {\eta} {\E {\eta_m} {\postexpz} \mid \Xi} \\
      &\leq \dlet m {\eta} {f'(\postexp(m))}
        & \text{(monotonicity of $\mathbb{E}$)} \\
      &\leq f'(\E {\eta} {\postexp})
        & \text{(Linearity of expectation)} \\
      &\leq f'(f(\preexp(m_1, m_2))).
        & \text{($f'$ is increasing)}
  \end{align*}

\item[\rname{Case}]
  Let $m_1, m_2 \models \pre$. We do a case analysis on
  $\dsem{m_1}{e_1}$ and conclude from we one of the two premises.

\item[\rname{Cond}]
  Immediate consequence of~\rname{Case} and~\rname{Struct}, using the
  synchronicity of both guards.

\item[\rname{SeqCase}]
  For $\ov{m} \models \post \land \exists i .\, {e_i}\lside$, we
  denote by $\iota(\ov{m})$ an index $i$ s.t.
  $\ov{m} \models {e_i}\lside$, and by $\eta_{\ov{m}}$ the coupling from
  \[ \Peprhl{s'_1}{s'_2}
        {\post \land {e_{\iota(\ov{m})}}\lside; \postexp}
        {\postz; \postexpz}{f_i} ,\]
  i.e. $\eta_{\ov{m}}$ is s.t.
  $\bcouplingsupp {\eta_{\ov{m}}} {\postexpz} {\delta_{\ov{m}}}
      {\dsem{\proj_1(\ov{m})}{s'_1}} {\dsem{\proj_2(\ov{m})}{s'_2}} {\post'}$,
  where $\delta_{\ov{m}} \eqdef f_{\iota(\ov{m})}(\postexp(\ov{m}))$.
  Let $m_1, m_2 \models \pre$ and $\mu$ s.t.
  $\bcouplingsupp \mu {\postexp} {\delta}
      {\dsem{m_1}{s_1}} {\dsem{m_2}{s_2}} {\post}$,
  where $\delta \eqdef f_0(\preexp(m_1, m_2))$---such
  a coupling is obtained from the premise
  $\Peprhl{s_1}{s_2}{\pre ; \preexp}{\post ; \postexp}{f_0}$.
  Let $\eta \eqdef \E {\ov{m} \sim \mu} {\eta_{\ov{m}}}$.
  The distribution $\eta$ is well-defined if for any
  $\ov{m} \in \supp(\mu)$,
  $\ov{m} \models \post \land \exists i .\, {e_i}\lside$.
  By definition of $\mu$, we already know that
  $\supp(\mu) \subseteq \post$.
  Moreover, from the premise
  $\post \implies {\textstyle\bigvee}_{i \in I} e_i$,
  we obtain the existence of a $\iota \in I$ s.t.
  $\proj_1(\ov{m}) \models e_\iota$, i.e. such that
  $\ov{m} \models {e_{\iota}}\lside$.
  It is immediate that $\supp(\eta) \subseteq \post'$ since for any
  $\ov{m} \in \supp(\mu)$, by definition of $\eta_{\ov{m}}$, we know
  that $\eta_{\ov{m}} \subseteq \post'$.
  Now, for $i \in \{1, 2\}$, we have:
  \begin{align*}
    \proj_i(\eta)
      &= \proj_i(\E {\ov{m} \sim \mu} {\eta_{\ov{m}}})
       = \E {\ov{m} \sim \mu}
           {\underbrace{\proj_i(\eta_{\ov{m}})}_{%
              \mathclap{\dsem{\proj_i(\ov{m})}{s'_i}}}} \\
      &= \E {m \sim \proj_i(\mu)} {\dsem{m}{s'_i}}
       = \E {m \sim \dsem{m_i}{s_i}} {\dsem{m}{s'_i}} \\
      &= m \mapsto \dsem{m}{s_i; s'_i}.
  \end{align*}

  We are left to prove the bounding property of $\eta$. For $i \in I$,
  we denote by $\ov{p}_i$ the quantity
  ${\textstyle\Pr}_{\ov{m} \sim \mu} [\iota(\ov{m}) = i]$.
  Then,
  \[
    \ov{p}_i
      =    {\textstyle\Pr}_{\ov{m} \sim \mu} [\iota(\ov{m}) = i]
      \leq {\textstyle\Pr}_{\ov{m} \sim \mu} [\dsem{\proj_1(\ov{m})}{e_i}]
      =    {\textstyle\Pr}_{
                \underbrace{\scriptstyle m \sim \proj_1(\mu)}_{%
                  m \sim \dsem{m_1}{s_1}}}
            [ \dsem{m}{e_i} ] .
  \]
  Denote this last quantity by $p_i$.  By the law of total expectation:
  \begin{align*}
    \E \mu  {\postexpz}
      &= \E {\ov{m} \sim \mu} {\E {\mu_{\ov{m}}} \postexpz} \\
      &= \sum_{i \in I} \ov{p}_i \cdot
           \E {\ov{m} \sim \mu} {%
             \E {\mu_{\ov{m}}} {\postexpz} \mid \iota(\ov{m}) = i} \\
      &\leq \sum_{i \in I} p_i \cdot
           \E {\ov{m} \sim \mu} {%
             \E {\mu_{\ov{m}}} {\postexpz} \mid \iota(\ov{m}) = i}.
  \end{align*}

  Now, for $\ov{m} \in \supp(\mu)$ s.t. $\iota(\ov{m}) = i$, we have:
  \begin{align*}
    \E {\mu_{\ov{m}}} \postexpz
      & \leq \delta_{\ov{m}} = f_i (\postexp(\ov{m})).
  \end{align*}
  Hence,
  \begin{align*}
    \E \mu  {\postexpz}
      &\leq \sum_{i \in I} p_i \cdot \E {\ov{m} \sim \mu}
              {f_i (\postexp(\ov{m})) \mid \iota(m) = i} \\
      &\leq \sum_{i \in I} p_i \cdot \E {\ov{m} \sim \mu}
              {f_i (\postexp(\ov{m}))}
       = \sum_{i \in I} p_i f_i(\E {\mu} {\postexp}) \\
      &\leq \sum_{i \in I} p_i \cdot f_i(f_0(\preexp(m_1, m_2)))
       = \ov{f}(\preexp(m_1, m_2))
  \end{align*}
  where the last step is by the premise.

\item[\rname{While}]
  We proceed by induction on $n$.
  For $i \in \{1, 2\}$, let $\ov{s_i} \eqdef \iwhile e {s_i}$.
  For $n \in \NN$, let $\post_n \eqdef \post \land (i\lside = n)$
  and $\ov{f}_n \eqdef f_1 \circ \cdots \circ f_n$.
  If $n=0$, under $m_1, m_2 \models \post$, we have
  $\dsem{m_1}{e} = \dsem{m_2}{e} = \false$.
  Hence, for $i \in \{1, 2\}$,
  $\dsem{m_i}{\ov{s_i}} = \dunit{m_i}$
  and we are in a case similar to \rname{Skip}.
  Otherwise, assume that the rule is valid for $n$.
  From the premises and the induction hypothesis, we have:
  \begin{align*}
    & \Peprhl {s_1} {s_2}
      {\post_{n+1} \land {e_1}\lside; \postexp_{n + 1}}
      {\post_n; \postexp_n}{f_{n+1}} \\
    & \Peprhl {\ov{s_1}} {\ov{s_2}}
      {\post_n; \postexp_n} {\post_0; \postexp_0} {\ov{f}_n}
  \end{align*}

  By reasoning similar to \rname{Seq}, we have
  $\Peprhl {s_1; \ov{s_1}} {s_2; \ov{s_2}}
      {\post_{n+1}; \postexp_{n + 1}} {\post_0; \postexp_0} {\ov{f}_{n+1}}.$
  Now, under $m_1, m_2 \models \post$, we have, for $i \in \{1, 2\}$,
  $\dsem{m_i}{s_i; \ov{s_i}} = \dsem{m_i}{\ov{s_i}}$.
  Hence, by reasoning similar to the one of~\rname{Struct},
  we obtain
  $\Peprhl {\ov{s_1}} {\ov{s_2}}
      {\post_{n+1} \land {e_1}\lside; \postexp_{n + 1}} {\post_0; \postexp_0}
      {\ov{f}_{n+1}}.$
  By $\post_{n+1} \iff (\post_{n+1} \land {e_1}\lside)$,
  we conclude that
  $\Peprhl {\ov{s_1}} {\ov{s_2}}
      {\post_{n+1}; \postexp_{n + 1}} {\post_0; \postexp_0}
      {\ov{f}_{n+1}}.$

\item[\rname{Frame-D}]
  Let $m_1, m_2 \models \pre$. From the premise, there is a coupling $\eta$ s.t.
  $\bcouplingsupp {\eta} {\postexp} {\delta}
     {\dsem{m_1}{s_1}} {\dsem{m_2}{s_2}} {\post}$,
  where $\delta \eqdef f(\preexp(m_1, m_2))$.
  Now, we have
  \begin{align*}
    \E \eta {\postexp + \postexp'}
      &= \E \eta {\postexp} + \E \eta {\postexp'}
       \leq f(\preexp(m_1, m_2)) + \E \eta {\postexp'}.
  \end{align*}

  For $\ov{m}_1, \ov{m}_2 \in \supp(\eta)$, from $\proj_i(\eta) =
  \dsem{m_i}{s_i}$ and $\postexp' \# \MV(s_1), \MV(s_2)$, we have
  $\postexp'(\ov{m}_1, \ov{m}_2) = \postexp'(m_1, m_2)$.
  The last line is because $f$ is a non-contractive linear function, $f \in
  \mathcal{L}^\geq$. Hence,
  \begin{align*}
    \E \eta {\postexp + \postexp'}
      &\leq f(\preexp(m_1, m_2)) +
          \E {\ov{m}_1, \ov{m}_2 \sim \eta} {\postexp'(m_1, m_2)} \\
      &= f(\preexp(m_1, m_2)) + \wt{\eta} \cdot \postexp'(m_1, m_2)
       \leq f(\preexp(m_1, m_2)) + \postexp'(m_1, m_2) \\
      &\leq f(\preexp(m_1, m_2)) + f(\postexp'(m_1, m_2))
       = f(\preexp(m_1, m_2) + \postexp'(m_1, m_2)).
  \end{align*}

  Hence, $\eta$ is a coupling s.t.
  $\bcouplingsupp {\eta} {\postexp + \postexp'}
     {\delta'} {\dsem{m_1}{s_1}} {\dsem{m_2}{s_2}} {\post}$,
  where $\delta' \eqdef f((\preexp + \postexp')(m_1, m_2))$.

\item[\rname{Mult-Max}] A basic result about couplings is that for any two
  distributions $\eta_1, \eta_2$ over the same set, there exists a coupling
  $\eta$ such that:
  \[
  \Pr_{(a_1, a_2) \sim \eta} [ a_1 \neq a_2 ] = \text{TV}(\eta_1, \eta_2) .
  \]
  This coupling is called the \emph{maximal} or \emph{optimal} coupling (see,
  e.g., \citet{Thorisson00}).
  
  To show soundness of the rule, let $(m_1, m_2)$ two memories and, for $i \in
  \{1, 2\}$, let $\mu_i \eqdef \dsem{m_i}{\vec{x}\irnd \multD(\vec{p})}$.  Let
  $\nu_i$ be the distributions $\dsem{m_i}{\multD(\vec{p})}$.
  Let $\mu$ be a coupling of $\mu_1$ and $\mu_2$ such that the projection of
  $\mu$ on the variables $x\lside$ and $x\rside$ is a maximal coupling of
  $\nu_1$ and $\nu_2$; note that the projection of $\mu_1$ onto $x\lside$ is
  $\nu_1$, and the projection of $\mu_2$ onto $x\rside$ is $\nu_2$.
  Now, we can prove the inequality on distances:
  \[
    \E {(m_1', m_2') \sim \mu} {\| \dsem{m_1'}{\vec{x}} - \dsem{m_2'}{\vec{x}} \|_1}
    \leq \| \dsem{m_1}{\vec{p}} - \dsem{m_2}{\vec{p}} \|_1 .
  \]
  %  %
  By definition we have:
  %  %
  \begin{align}
    \E {(m_1', m_2') \sim \mu} {\| \dsem{m_1'}{\vec{x}} - \dsem{m_2'}{\vec{x}}\|_1}
    &= \sum_{m_1', m_2'} \mu(m_1', m_2') \cdot \| \dsem{m_1'}{\vec{x}} - \dsem{m_2'}{\vec{x}}\|_1
    \notag \\
    & =2 \sum_{m_1', m_2'} \sum_{a \neq b}
    \ind{ \dsem{m_1'}{\vec{x}} = a} \ind{ \dsem{m_2'}{\vec{x}} = b} \mu(m_1', m_2') 
    \tag{distance is $0$ or $2$} \\
    &= 2 \sum_{a \neq b} \sum_{m_1', m_2'}
    \ind{ \dsem{m_1'}{\vec{x}} = a} \ind{ \dsem{m_2'}{\vec{x}} = b} \mu(m_1', m_2')
    \notag \\
    &= 2 \sum_{a \neq b}  \Pr_{(m_1', m_2') \sim \mu} [ \dsem{m_1'}{\vec{x}} = a, \dsem{m_2'}{\vec{x}} = b ]
    \notag \\
    &= 2\cdot\Pr_{(m_1', m_2') \sim \mu} [ \dsem{m_1'}{\vec{x}} \neq \dsem{m_2'}{\vec{x}} ]
    \notag \\
    &= 2 \text{TV}(\nu_1, \nu_2)
    \tag{maximal coupling} \\
    &= \| \dsem{m_1}{\vec{p}} - \dsem{m_2}{\vec{p}} \|_{1} .
    \notag
  \end{align}

\item[\rname{Trans}]
  We prove by induction on $n \in \NN$ that for every two memories $m_1$ and
  $m_2$ such that $\preexp(m_1, m_2) = n$, there exists a coupling $\mu$ such
  that
  \[
    \bcouplingsupp \mu {\postexp} {f(\preexp(m_1, m_2))}
    {\dsem{m_1}{s}} {\dsem{m_2}{s}} {\post} .
  \]
  For the base case $\preexp(m_1, m_2) = 0$, the inductive hypothesis on the
  premise
  \[
    \Peprhl{s}{s}{\pre \land \preexp=0;-}{\post;\postexp}{0}
  \]
  give the desired coupling.

  For the inductive step $\preexp(m_1, m_2) = n+1$, by path
  compatibility $\pathcompat(\pre, \preexp)$  there exists $m'$ with
  $\preexp(m_1, m') = 1$ and $\preexp(m', m_2) = n$ such that $\pre(m_1, m')$
  and $\pre(m', m_2)$. By induction on $n$, there exists $\mu_n$ such that
  $\bcouplingsupp {\mu_n} {\postexp}
      {\delta_n} {\dsem{m'}{s}} {\dsem{m_2}{s}} {\post^*}$,
  where $\delta_n \eqdef f(\preexp(m', m_2))$.
  By the inductive hypothesis on premise
  \[
  \Peprhl{s}{s}{\pre\land\preexp=1; -}{\post; \postexp}{\const{f(1)}} ,
\]
  there exists $\mu_1$ such that
  $\bcouplingsupp {\mu_1} {\postexp} {f(1)}
     {\dsem{m_1}{s}} {\dsem{m'}{s}} {\post}$,
  Define the coupling
  \begin{align*}
    \mu(m_1, m_2) \triangleq
    \sum_{m} \frac{\mu_1(m_1, m) \cdot \mu_n(m, m_2)}{M(m)}
  \end{align*}
  where $M(m) \eqdef \pi_1(\mu_n)(m) = \pi_2(\mu_1)(m) = \dsem{m'}{s}$ and we
  drop terms with $M = 0$. By induction, $\supp(\mu) \models {\post}^2 \subseteq
  \post^*\subseteq \post$ since $\post$ is transitive. The marginal conditions
  are straightforward: for any $m_1$,
  \begin{align*}
    \pi_1(\mu)(m_1)
      &= \sum_{m_2} \mu (m_1, m_2) \\
      &= \sum_{m} \Bigg(
           \frac{\mu_1(m_1, m)}{\pi_1(\mu_n)(m)} \cdot
           \underbrace{%
             \sum_{m_2} \mu_n(m, m_2)}_{\pi_1(\mu_n)(m)}
         \Bigg) \\
        &= \sum_{m} \mu_1(m_1, m) = \pi_1(\mu_1)(m_1) = \dsem{m_1}{s} .
  \end{align*}
  Similarly, $\pi_2(\mu) = \dsem{m_2}{s}$.  Finally, we show the expected distance
  condition:
  \begin{align*}
    \E {\mu } {\postexp}
      &= \sum_{m_1, m_2} \mu (m_1, m_2) \cdot \postexp(m_1, m_2) \\
      &= \sum_{m_1, m_2, m}
           \frac{\mu_1(m_1, m) \cdot \mu_n(m, m_2)}{M(m)} \cdot \postexp(m_1, m_2) \\
      &\leq \sum_{m_1, m}
              \mu_1(m_1, m) \cdot \postexp(m_1, m)
                \sum_{m_2} \frac{\mu_n(m, m_2)}{M(m)} \\
      &\quad \quad \quad + \sum_{m, m_2}
            \mu_n(m, m_2) \cdot \postexp(m, m_2)
                \sum_{m_1} \frac{\mu_1(m_1, m)}{M(m)} \tag{$\postexp$ satisfies (H)} \\
      &= \E {\mu_1} {\postexp} + \E {\mu_n} {\postexp} \\
      &\leq f(\preexp(m_1, m')) + f(\preexp(m', m_2)) \tag{induction hypotheses} \\
      &= f (\preexp(m_1, m_2))  \tag{$f \in \mathcal{L}$} .
  \end{align*}

\end{description}
\end{proof}

\section{Details for Examples}

\subsection{Convex SGM (\cref{ex:convex-sgm})}

We detail the bounds in the two cases. In the first case, the selected samples
$S[i]\lside$ and $S[i]\rside$ may be different. We need to show:
\[
  \| (w\lside - \alpha_t \cdot (\nabla \ell(S[i], -))(w)\lside)
  - (w\rside - \alpha_t \cdot (\nabla \ell(S[i], -))(w)\rside) \|
  \leq \| w\lside - w\rside \| + 2 \alpha_t L .
\]
We can directly bound:
\begin{align*}
  &\| (w\lside - \alpha_t \cdot (\nabla \ell(S[i], -))(w)\lside)
  - (w\rside - \alpha_t \cdot (\nabla \ell(S[i], -))(w)\rside) \| \\
  &\leq \| w\lside - w\rside \|
  + \alpha_t \| (\nabla \ell(S[i], -))(w)\lside \|
  + \alpha_t \| (\nabla \ell(S[i], -))(w)\rside \| \\
  &\leq \| w\lside - w\rside \| + 2 \alpha_t L
\end{align*}
where the first inequality is by the triangle inequality, and the second
follows since $\ell(z, -)$ is $L$-Lipschitz. Thus, we can take
$f = \trans{2 \alpha_t L}$ in the first case.

The second case boils down to showing
\[
  \| (w\lside - \alpha_t \cdot (\nabla \ell(S[i], -))(w)\lside)
  - (w\rside - \alpha_t \cdot (\nabla \ell(S[i], -))(w)\rside) \|
  \leq \| w\lside - w\rside \| .
\]
when $S[i]\lside = S[i]\rside$. This follows from a calculation similar to the
proof by \citet[Lemma 3.7.2]{HardtRS16}:
\begin{align*}
  &\| (w\lside - \alpha_t \cdot (\nabla \ell(S[i], -))(w)\lside)
  - (w\rside - \alpha_t \cdot (\nabla \ell(S[i], -))(w)\rside) \|^2 \\
  &= \| w\lside - w\rside \|^2
  - 2 \alpha_t \langle (\nabla \ell(S[i], -))(w)\lside - (\nabla \ell(S[i], -))(w)\rside,
  w\lside - w\rside \rangle \\
  &+ \alpha_t^2\| (\nabla \ell(S[i], -))(w)\lside) - (\nabla \ell(S[i], -))(w)\rside \|^2 \\
  &\leq \| w\lside - w\rside \|^2
  - (2 \alpha_t/\beta - \alpha_t^2)
  \| (\nabla \ell(S[i], -))(w)\lside) - (\nabla \ell(S[i], -))(w)\rside \|^2 \\
  &\leq \| w\lside - w\rside \|^2 .
\end{align*}
The first inequality follows since convexity and Lipschitz gradient implies
that
\[
  \langle (\nabla \ell(S[i], -))(w)\lside - (\nabla \ell(S[i], -))(w)\rside,
  w\lside - w\rside \rangle
  \geq \frac{1}{\beta} \| (\nabla \ell(S[i], -))(w)\lside - (\nabla \ell( S[i],
  -))(w)\rside \|^2 .
\]
The second inequality follows from $0 \leq \alpha_t \leq 2/\beta$. Thus, we can
take $f = \id$ in the second case.

\subsection{Non-Convex SGM (\cref{ex:nonconvex-sgm})}

Suppose that the loss function $\ell$ is bounded in $[0, 1]$, possibly
non-convex, but $L$-Lipschitz and with $\beta$-Lipschitz gradient.  Suppose that
we take non-increasing step sizes $0 \leq \alpha_t \leq \sigma/t$ for some
constant $\sigma \geq 0$.  Then, we will prove the following judgment:
\[
  \Peprhl{\sgm}{\sgm}
  { \Adj(S\lside, S\rside); - }
  { \top; | \ell(w\lside, z) - \ell(w\rside, z) | }
  { \const{\epsilon} }
\]
where
\[
  \epsilon \triangleq (2/n) \left\lceil
    \left( \frac{2L^2}{\beta (1 - 1/n)} \right)^{1/(q + 1)} T^{q/(q + 1)}
  \right\rceil .
\]

This example uses an advanced analysis from \citet[Lemma 3.11]{HardtRS16}.  We
can't directly express that result in our logic, but we can inline the proof.
Roughly, the idea is that with large probability, the first bunch of steps don't
see the differing example. By the time we hit the differing example, the step
size has already decayed enough. To model this kind of reasoning, we will use
the program transformation rules to split the loop into iterations before the
critical step, and iterations after the critical step.  Then, we will perform a
probabilistic case in between, casing on whether we have seen the differing
example or not.

To begin, let the critical iteration be
\[
  t_0 \triangleq \left\lceil
    \left( \frac{2L^2}{\beta (1 - 1/n)} \right)^{1/(q + 1)} T^{q/(q + 1)}
  \right\rceil
\]
where $q \triangleq \beta \sigma$.  We can split the loop in $\sgm$ into two:
\[
  \begin{array}{l}
    t \iass 0; \\
    \iwhile{t < T \land t < t_0}{} \\
    \quad i \irnd [n]; \\
    \quad w \iass w - \alpha_t \cdot (\nabla \ell(S[i], -)) (w); \\
    \quad t \iass t + 1; \\
    \iwhile{t < T}{} \\
    \quad i \irnd [n]; \\
    \quad w \iass w - \alpha_t \cdot (\nabla \ell(S[i], -)) (w); \\
    \quad t \iass t + 1; \\
    \iret{w}
  \end{array}
\]
Call the loops $c_<$ and $c_\geq$, with loop bodies $w_<$ and $w_\geq$. In
the first loop, we will bound the probability of $\| w\lside - w\rside \| >
0$. We want to prove the judgment
\[
  \eprhl{w_<}{w_<}
  {t\lside = t\rside ; \ind{w \lside \neq w\rside} }
  {t\lside = t\rside ; \ind{w \lside \neq w\rside} }
  { \trans{1/n} } .
\]
Again, we use the identity coupling when sampling $i$. Then, we case on
whether we hit the differing example or not. In the first case, we hit the
differing example and we need to prove
\[
  \eprhl{w_<}{w_<}
  {t\lside = t\rside ; \ind{w \lside \neq w\rside} }
  {t\lside = t\rside ; \ind{w \lside \neq w\rside} }
  { \trans{1} } .
\]
This boils down to showing:
\[
  \ind{w - \alpha_t \cdot (\nabla \ell(S[i], -)) (w) \lside
    \neq w - \alpha_t \cdot (\nabla \ell(S[i], -)) (w) \rside}
  \leq \ind{w\lside \neq w\rside} + 1
\]
but this is clear since the indicator is in $\{ 0, 1 \}$.

In the second case, we hit the same example and need to prove:
\[
  \eprhl{w_<}{w_<}
  {t\lside = t\rside \land S[i] \lside = S[i] \rside ; \ind{w \lside \neq w\rside} }
  {t\lside = t\rside ; \ind{w \lside \neq w\rside} }
  { \id } .
\]
This boils down to showing:
\[
  \ind{w - \alpha_t \cdot (\nabla \ell(S[i], -)) (w) \lside
    \neq w - \alpha_t \cdot (\nabla \ell(S[i], -)) (w) \rside}
  \leq \ind{w\lside \neq w\rside}
\]
assuming that $S[i] \lside = S[i] \rside$. But this is clear also---if
$w\lside \neq w\rside$ then there is nothing to prove, otherwise if
$w\lside = w\rside$ then the projections are equal.

Putting these two cases together (noting that they happen with probability
$1/n$ and $1 - 1/n$ respectively) and applying the loop rule, we have:
\[
  \eprhl{w_<}{w_<}
  {t\lside = t\rside ; \ind{w \lside \neq w\rside} }
  {t\lside = t\rside ; \ind{w \lside \neq w\rside} }
  { \trans{t_0/n} }
\]
as desired.

Now, we perform a probabilistic case on $w\lside = w\rside$. Suppose $w\lside =
w\rside$. In the second loop, we know that $t\lside = t\rside \geq t_0$.  By
similar reasoning to the previous sections, we have:
\[
  \eprhl{w_\geq}{w_\geq}
  { t\lside = t\rside \land t\lside \geq t_0 ; \| w\lside - w\rside \| }
  { t\lside = t\rside \land t\lside \geq t_0 ; \| w\lside - w\rside \| }
  { f_c }
\]
where
\begin{align*}
  f_c(x) &\triangleq (1/n + (1 - 1/n)(1 + \alpha_t \beta)) x + 2 \alpha_t L / n \\
  &\leq (1 + (1 - 1/n)\sigma \beta/t) x + 2 \sigma L / t n \\
  &\leq \exp((1 - 1/n)\sigma \beta/t) x + 2 \sigma L / t n .
\end{align*}
In the last step, we use $1 + x \leq \exp(x)$.

We can then apply the loop rule to show:
\[
  \eprhl{c_\geq}{c_\geq}
  { t\lside = t\rside \land t\lside \geq t_0 \land w\lside = w\rside ; \| w\lside - w\rside \| }
  { t\lside = t\rside \land t\lside \geq t_0 ; \| w\lside - w\rside \| }
  { f }
\]
where
\begin{align*}
  f(x) &\triangleq
  x \cdot \prod_{r = t_0 + 1}^T \exp \left( (1 - 1/n) \frac{\sigma \beta}{r} \right)
  + \sum_{s = t_0 + 1}^T \frac{2\sigma L}{s n} \prod_{r = s + 1}^T
  \exp \left( (1 - 1/n) \frac{\sigma \beta}{r} \right) \\
  &= x \cdot \exp \left( \sigma \beta (1 - 1/n) \sum_{r = t_0 + 1}^T \frac{1}{r} \right)
  + \sum_{s = t_0 + 1}^T \frac{2\sigma L}{s n}
  \exp \left( \sigma \beta (1 - 1/n) \sum_{r = s + 1}^T \frac{1}{r} \right) \\
  &\leq x \cdot \exp \left( \sigma \beta (1 - 1/n) \log(T/t_0) \right)
  + \sum_{s = t_0 + 1}^T \frac{2\sigma L}{s n}
  \exp \left( \sigma \beta (1 - 1/n) \log(T/s) \right) \\
  &= x \cdot \exp \left( \sigma \beta (1 - 1/n) \log(T/t_0) \right)
  + \frac{2\sigma L}{n} T^{\beta \sigma(1 - 1/n)}
  \sum_{s = t_0 + 1}^T s^{- \beta \sigma(1 - 1/n) - 1} \\
  &\leq x \cdot \exp \left( \sigma \beta (1 - 1/n) \log(T/t_0) \right)
  + \frac{2\sigma L}{n} T^{\beta \sigma(1 - 1/n)}
  \cdot \frac{1}{\beta \sigma (1 - 1/n)} t_0^{-\beta \sigma (1 - 1/n)}
  \\
  &= x \cdot \exp \left( \sigma \beta (1 - 1/n) \log(T/t_0) \right)
  + \frac{2 L}{\beta(n - 1)} \left(\frac{T}{t_0}\right)^{\beta \sigma(1 - 1/n)}
  \\
  &\leq x \cdot \exp \left( \sigma \beta (1 - 1/n) \log(T/t_0) \right)
  + \frac{2 L}{\beta(n - 1)} \left(\frac{T}{t_0}\right)^{\beta \sigma} .
\end{align*}
Let the last term be $\rho$.  The first inequality uses $\sum_{t = a +
  1}^b 1/t \leq \log(b/a)$ and the second inequality uses $\sum_{t =
  a + 1}^b 1/t^c \leq a^{1 - c} / (c - 1)$ for $c > 1$; both facts
follow from bounding the sum by an integral. By applying the Lipschitz
assumption on $\ell$ and the \rname{Conseq} rule, we have:
\[
  \eprhl{c_\geq}{c_\geq}
  { t\lside = t\rside \land t\lside \geq t_0 \land w\lside = w\rside ; - }
  { t\lside = t\rside \land t\lside \geq t_0 ; | \ell(w, z)\lside - \ell(w,
    z)\rside |}
  { \const{L\rho} }
\]
for every example $z \in Z$.

In the other case, suppose $w\lside \neq w\rside$. Applying the rule of
consequence using the fact that the loss function is bounded in $[0, 1]$, we
have:
\[
  \eprhl{c_\geq}{c_\geq}
  { t\lside = t\rside \land t\lside \geq t_0 \land w\lside \neq w\rside ; - }
  { t\lside = t\rside \land t\lside \geq t_0 ; | \ell(w, z)\lside - \ell(w, z)\rside| }
  { \const{1} } .
\]
Applying the rule \rname{SeqCase-A} to link the two loops, we have:
\[
  \eprhl{\sgm}{\sgm}
  { \Adj(S\lside, S\rside) ; - }
  { \top ; | \ell(w, z)\lside - \ell(w, z)\rside| }
  { \const{t_0/n + L\rho} } .
\]
Note that setting
\[
  t_0 \geq \delta \triangleq \left( \frac{2L^2}{\beta (1 - 1/n)} \right)^{1/(q + 1)}
  T^{q/(q + 1)}
\]
gives $t_0 / n + L \rho \leq 2 t_0 / n$ since $\delta$ balances the two terms,
so we can conclude.

The proof uses an advanced sequential composition rule \rname{SeqCase-A}, shown
in \cref{fig:morerules}. This rule combines sequential composition with a case
analysis on an event that may depend on both memories.

\begin{figure}[h!b]
\[
\inferrule[SeqCase-A]
{ 
   \Peprhl{s_1}{s_2}{\pre;-}{\Theta; \ind{e}}{\const{\gamma}} \\\\
   \Peprhl{s'_1}{s'_2}{\Theta \land e; -}{\post; \preexp}{f} \\
   \Peprhl{s'_1}{s'_2}{\Theta \land \neg e; -}{\post; \preexp}{f_\neg} }
 { \Peprhl{s_1;s'_1}{s_2;s'_2}
     {\pre; -}{\post; \preexp}{\const{\gamma} \cdot f + f_\neg} }
\]
\caption{\label{fig:morerules} Advanced sequential case rule}
\end{figure}

\end{document}

\fi

\end{document}

%%% Local Variables:
%%% mode: latex
%%% TeX-master: t
%%% End:

%% file: prelude.tex
\usepackage{subfiles}

\usepackage{bm}

\usepackage{todonotes}

\usepackage{cleveref}
\crefname{section}{\S}{\S}
\Crefname{section}{\S}{\S}

\usepackage{mathpartir}

\newcommand{\rname}[1]{[\textsc{#1}]}

\let\xinferrule\inferrule
\renewcommand{\inferrule}[3][]{%
  \ifx&#1&
    \xinferrule*{#2}{#3}
  \else
    \xinferrule*[left=\rname{#1}]{#2}{#3}
  \fi}

\makeatletter
\def\arcr{\@arraycr}
\makeatother

\theoremstyle{acmdefinition}
\newtheorem*{fact*}{Fact}
\newtheorem*{definition*}{Definition}

\newcommand{\distr}{\mathbb{D}}

\newcommand{\BB}{\mathbb{B}}
\newcommand{\NN}{\mathbb{N}}
\newcommand{\RR}{\mathbb{R}}

\newcommand{\ov}[1]{\overline{#1}}

\newcommand{\coupling}[3]
{{#2} \mathrel{\langle #1 \rangle} {#3}}

\newcommand{\bcoupling}[5]
{{#4} \mathrel{\langle {#1} \rangle_{#2 \leq #3}} {#5}}
\newcommand{\bcouplingsupp}[6]
{{#4} \mathrel{\langle {#1} \rangle_{#2 \leq #3}^{#6}} {#5}}

\DeclareMathOperator{\supp}{supp}
\DeclareMathOperator{\id}{id}
\DeclareMathOperator{\MV}{MV}

\newcommand{\wt}[1]{|{#1}|}

\newcommand{\dunit}[2][]{{{\delta}^{#1}_{#2}}}

\newcommand{\Exp}{\mathbb{E}}

\newcommand{\E}[2]{\Exp_{#1} [{#2}]}

\newcommand{\bernD}{\mathbf{Bern}}
\newcommand{\multD}{\mathbf{Mult}}

\newcommand{\dlet}[3]{\Exp_{{#1} \sim {#2}} [{#3}]}
\newcommand{\dslet}[2]{\Exp_{#1} [{#2}]}

\def\eqdef{\triangleq}

\newcommand{\proj}{\pi}

\newcommand{\kw}[1]{{\ensuremath{\mathsf{#1}}}\xspace}

\newcommand{\kskip}{\kw{skip}}
\newcommand{\kabort}{\kw{abort}}
\newcommand{\kif}{\kw{if}}
\newcommand{\kthen}{\kw{then}}
\newcommand{\kelse}{\kw{else}}
\newcommand{\kwhile}{\kw{while}}
\newcommand{\kdo}{\kw{do}}
\newcommand{\kreturn}{\kw{return}}

\newcommand{\rndarrow}{%
  \stackrel{\raisebox{-.25ex}[.25ex]%
   {\tiny $\mathdollar$}}{\raisebox{-.25ex}[.25ex]{$\leftarrow$}}}

\newcommand{\tobij}{%
  \stackrel{\raisebox{-.1ex}[.25ex]%
   {\scriptsize 1-1}}{\raisebox{-.1ex}[.25ex]{$\longrightarrow$}}}

\newcommand{\false}[0]{\mathop{\perp}}
\newcommand{\iabort}[0]{\kabort}
\newcommand{\iskip}[0]{\kskip}
\newcommand{\iass}[0]{\mathrel{\leftarrow}}
\newcommand{\irnd}[0]{\mathrel{\rndarrow}}
\newcommand{\ifte}[3]{\kif\ {#1}\ \kthen\ {#2}\ \kelse\ {#3}}
\newcommand{\ift}[2]{\kif\ {#1}\ \kthen\ {#2}}

\newcommand{\iwhile}[2]{\kwhile\ {#1}\ \kdo\ {#2}}
\newcommand{\iret}[1]{\kreturn\ {#1}}

\newcommand{\mem}{{\mathcal{M}}}

\newcommand{\vars}{{\mathcal{V}}}
\newcommand{\exprs}{{\mathcal{E}}}
\newcommand{\dexprs}{{\mathcal{D}}}

\def\ms{\mspace{-1.5mu}}

\newcommand{\lmark}{{{}_{\scriptscriptstyle \vartriangleleft}}}
\newcommand{\rmark}{{{}_{\scriptscriptstyle \vartriangleright}}}

\newcommand{\lside}{_{\ms\lmark}}
\newcommand{\rside}{_{\ms\rmark}}

\newcommand{\sem}[1]{{\llbracket {#1} \rrbracket}}
\newcommand{\dsem}[2]{\sem{#2}_{#1}}

\newcommand{\subst}[2]{{#1}\coloneq{#2}}
\newcommand{\ind}[1]{\mathbb{1}[{#1}]}

\newcommand{\pwhile}{\textsc{pWhile}\xspace}

\newcommand{\Sprhl}{\textsc{pRHL}\xspace}
\newcommand{\Saprhl}{\textsc{apRHL}\xspace}
\newcommand{\Spgcl}{\textsc{pGCL}\xspace}
\newcommand{\Sppdl}{\textsc{PPDL}\xspace}

\newcommand{\SYSTEM}{\texorpdfstring{$\mathbb{E}$\textsc{pRHL}}{EpRHL}\xspace}
\newcommand{\Seprhl}{\SYSTEM}
\newcommand{\Sxprhl}{$\times$\textsc{pRHL}\xspace}
\newcommand{\eqsem}[3]{#1 \vdash #2 \equiv #3}

\newcommand{\Peprhl}[5]{\vdash \{ #3 \}\ {#1} \sim_{#5} {#2}\ \{ #4 \}}

\newcommand{\eprhl}[5]{\{ #3 \}\ {#1} \sim_{#5} {#2}\ \{ #4 \}}

\newcommand{\pre}{\Phi}
\newcommand{\post}{\Psi}
\newcommand{\postz}{\Psi'}
\newcommand{\preexp}{\mathfrak{d}}
\newcommand{\postexp}{\mathfrak{d}'}
\newcommand{\postexpz}{\mathfrak{d}''}

\newcommand{\const}[1]{\bm{#1}}
\newcommand{\scale}[1]{\bullet{#1}}
\newcommand{\trans}[1]{+{#1}}

\newcommand{\wpc}{\mathsf{wp}}

\newcommand{\pathdistance}[1]{\text{pd}_{#1}}

\newcommand{\Adj}{\mathrm{Adj}}

\newcommand{\pathcompat}{\text{PathCompat}}